%% file: 1801a.tex
\documentclass[11pt]{article}
\usepackage{jheppub}

\include{1801}

\begin{document}
\title{Biadjoint scalar tree amplitudes\\and intersecting dual associahedra}
\author{Hadleigh Frost}
\affiliation{Mathematical Institute, University of Oxford, Woodstock Road, Oxford OX2 6GG, UK.}
\emailAdd{frost@maths.ox.ac.uk}

\abstract{We present a new formula for the biadjoint scalar tree amplitudes $m(\alpha|\beta)$ based on the combinatorics of dual associahedra. Our construction makes essential use of the cones in `kinematic space' introduced by Arkani-Hamed, Bai, He, and Yan. We then consider dual associahedra in `dual kinematic space.' If appropriately embedded, the intersections of these dual associahedra encode the amplitudes $m(\alpha|\beta)$. In fact, we encode all the partial amplitudes at $n$-points using a single object, a `fan,' in dual kinematic space. Equivalently, as a pleasant corollary of our construction, all $n$-point partial amplitudes can be understood as coming from integrals over subvarieties in a single toric variety. Explicit formulas for the amplitudes then follow by evaluating these integrals using the equivariant localisation (or `Duistermaat-Heckman') formula. Finally, by introducing a lattice in kinematic space, we observe that our fan is also related to the inverse KLT kernel, sometimes denoted $m_{\alpha'}(\alpha|\beta)$.}
\maketitle
\section{Introduction}
\label{introduction}
Arkani-Hamed, Bai, He, and Yan (AHBHY) presented new formulas for the biadjoint scalar tree amplitudes in ref. \cite{ahbhy17}. Their formulas are based on constructions in `kinematic space,' $\mc{K}_n$, which is the vector space of all Mandelstam variables $s_{ij} = 2k_i\cdot k_j$ subject to the momentum conservation relations. For example, at four points, the kinematic space $\mc{K}_4$ is the plane defined by
$$
s_{12}+s_{13}+s_{23} = 0
$$
in $\mb{R}^3$. AHBHY construct embeddings of the associahedron into kinematic space.\footnote{This is carried out in section 3.2 of ref. \cite{ahbhy17}.} These are polytopes in $\mc{K}_n$ of dimension $n-3$. For instance, at four points, the ordering $\alpha = 1234$ is associated to the line segment $\mc{A}(1234)$ between the two points
$$
(s_{12},s_{13},s_{23}) = (c,-c,0)\qquad\text{and}\qquad (s_{12},s_{13},s_{23}) = (0,-c,c),
$$
for some constant $c$. This line segment is an embedding of the associahedron on three letters into kinematic space, regarded as a hyperplane in $\mb{R}^3$. See figure \ref{ex1a} for an illustration of this. There are two other embeddings (also shown in figure \ref{ex1a}). $\mc{A}(3124)$ is the line segment between
$$
(s_{12},s_{13},s_{23}) = (c,0,-c)\qquad\text{and}\qquad (s_{12},s_{13},s_{23}) = (0,c,-c),
$$
and $\mc{A}(2314)$ is the segment between
$$
(s_{12},s_{13},s_{23}) = (-c,0,c)\qquad\text{and}\qquad (s_{12},s_{13},s_{23}) = (-c,c,0).
$$
Notice that the three line segments are not intersecting. In general, AHBHY's construction embeds $(n-1)!/2$ non-intersecting associahedra in $\mc{K}_n$. These associahedra can be realised as the intersection of hyperplanes with $(n-1)!/2$ cones that we call $\mc{C}(\alpha)$. For instance, at four points, there are three non-intersecting cones which we show in figure \ref{ex4a}. An important observation for the present paper is that the `dual cones' $\mc{C}(\alpha)^*$ in dual kinematic space do intersect. The purpose of this paper is to explain how the intersections of the dual cones are related to the biadjoint scalar tree amplitudes. For example, the cones $\mc{C}(1234)$ and $\mc{C}(2314)$ do not intersect in $\mc{K}_4$. But their dual cones, $\mc{C}(1234)^*$ and $\mc{C}(2314)^*$, intersect in the span of a vector $W_{23}$. See figure \ref{ex1b}. The vector $W_{23}$ is `dual' to the Mandelstam variable $-s_{23}$ in the sense that $W_{23}\cdot Z = -s_{23}$ for $Z$ a vector in $\mc{K}_4$. But $-s_{23}$ is the propagator for the partial amplitude
$$
m(1234|2314) = \frac{1}{-s_{23}}.
$$
\begin{figure}
\begin{center}
\includegraphics[scale=0.40]{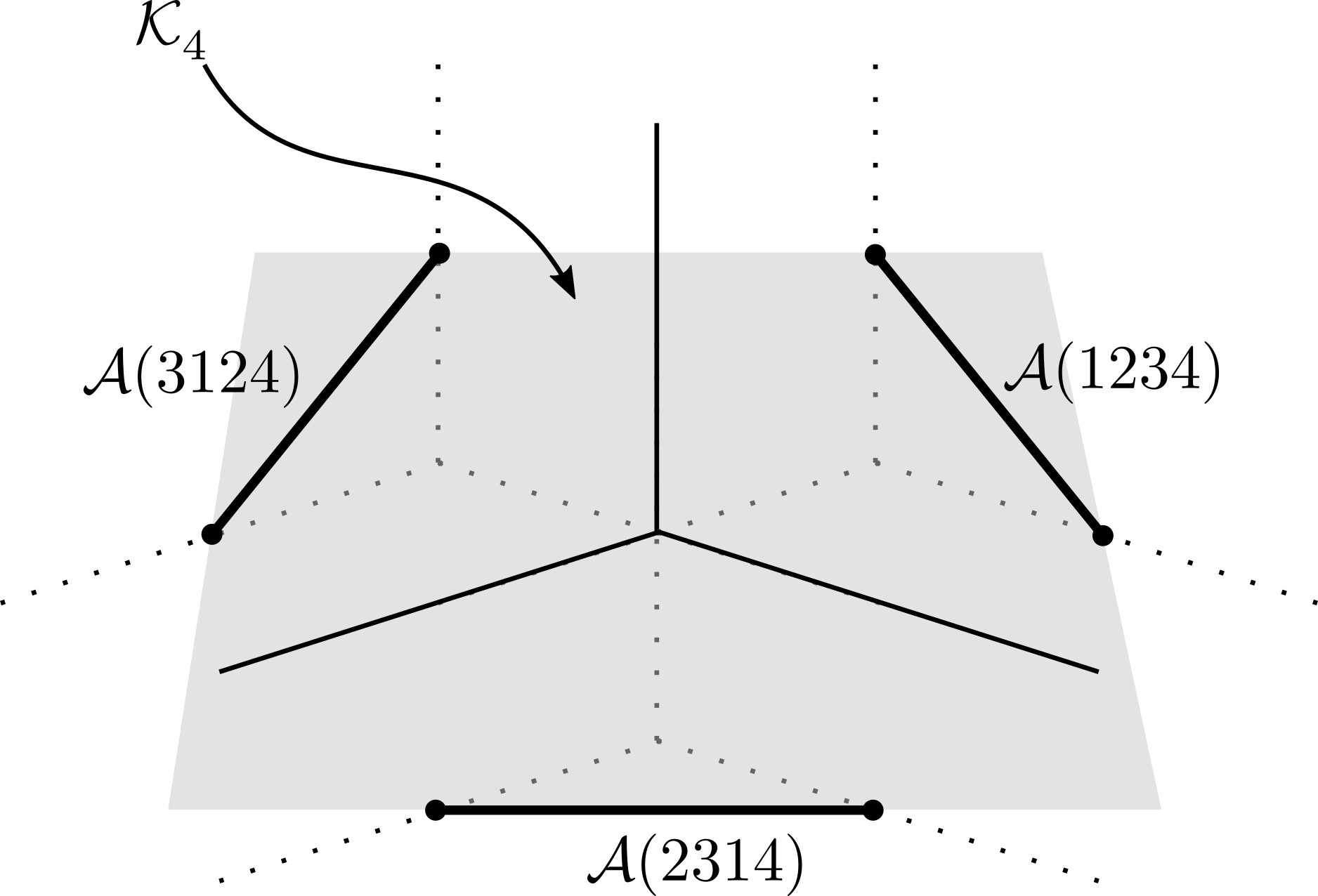}
\end{center}
\caption{The three associahedra $\mc{A}(\alpha)$ defined by AHBHY's construction at four points. The associahedra lie in a plane, kinematic space, inside $\mb{R}^3$.}
\label{ex1a}
\end{figure}
So this partial amplitude can also be written as
$$
m(1234|2314) = \frac{1}{W_{23}\cdot Z}.
$$
This suggests that the amplitude $m(1234|2314)$ is related to the intersection of the two dual cones $\mc{C}(1234)^*$ and $\mc{C}(2341)^*$. In this paper, we show how this observation can be made precise and generalised to all the partial tree amplitudes. Our main result is a generalisation of a formula given by AHBHY, which we now explain. Let $\alpha$ and $\beta$ be orderings of $\{1,..,n\}$. Then, following AHBHY, consider an associahedron $\mc{A}(\alpha)$ in kinematic space contained in some $n-3$ dimensional plane $H(\alpha) \subset \mc{K}_n$. The `dual polytope' $\mc{A}(\alpha)^*$ is a polytope which lives in the dual vector space $H(\alpha)^*$. Duality on polytopes maps dimension-$k$ faces in $\mc{A}(\alpha)$ to codimension-$k$ faces in $\mc{A}(\alpha)^*$. AHBHY prove the following formula,
\begin{equation}
\label{theirresult}
m(\alpha|\alpha) = \Vol (\mc{A}(\alpha)^*),
\end{equation}
where $\Vol$ denotes the `volume' of the polytope.\footnote{This formula is proposed in section 5.1 and explicitly proved in section 5.2 of ref. \cite{ahbhy17}, though it also follows from a general theorem concerning `positive geometries.'} The `volume' $\Vol$ is not a number, but rather a function on $H(\alpha)$. Or, if you prefer, the volume of the dual polytope depends on a choice of position (i.e. a choice of Mandelstam variables) $Z \in H(\alpha)$. The `volume' $\Vol$ is similar to an ordinary volume in the sense that it satisfies
$$
\Vol (A \cup B) = \Vol (A) + \Vol(B) - \Vol(A\cap B),
$$
for any two polytopes $A$ and $B$. We call anything which satisfies this relation a `valuation' on polytopes. Our main result can be stated as (see corollary \ref{valley}) 
\begin{equation}
\label{ourresult}
m(\alpha|\beta) = \Val (\p\mc{A}(\alpha)^*\cap \p\mc{A}(\beta)^*).
\end{equation}
In this formula, $\p\mc{A}(\alpha)^*$ is the boundary of $\mc{A}(\alpha)^*$ and `$\Val$' is a natural valuation on polytopes which we will introduce later. The novel aspect of this formula is that we must first define embeddings of the dual associahedra $\mc{A}(\alpha)^*$ into $\mc{K}_n^*$ (just as AHBHY define embeddings of $\mc{A}(\alpha)$ into $\mc{K}_n$). We define a `canonical embedding' in section \ref{embed}. Given this canonical embedding, it turns out that
$$
\Val (\p \mc{A}(\alpha)^*) = \Vol (\mc{A}(\alpha)^*).
$$
where, on the left-hand-side, we regard $\p\mc{A}(\alpha)^*$ as embedded in $\mc{K}_n^*$ and, on the right-hand-side, $\mc{A}(\alpha)^*$ is embedded only in $H(\alpha)^*$ (or some projective compactification of it, as done in AHBHY). For this reason, equation \eqref{ourresult} is a mild generalisation of AHBHY's result, equation \eqref{theirresult}. Before saying more about our results, we will recall some of the reasons to be interested in the biadjoint scalar amplitudes and their various presentations.
\begin{figure}
\begin{center}
\includegraphics[scale=0.40]{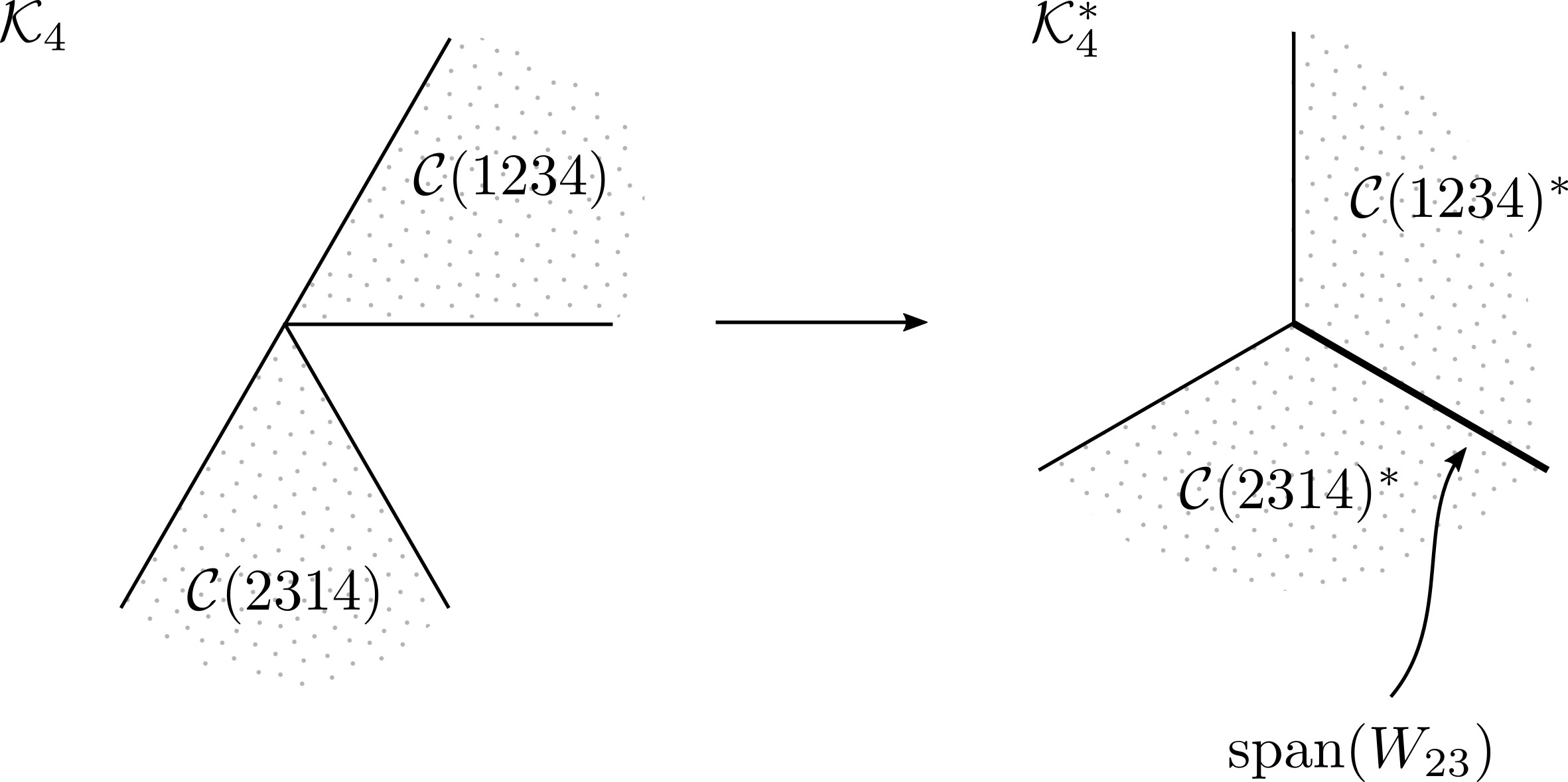}
\end{center}
\caption{The cones $\mc{C}(1234)$ and $\mc{C}(2314)$ do not intersect in kinematic space, $\mc{K}_4$. However, the dual cones $\mc{C}(1234)^*$ and $\mc{C}(2314)^*$ do intersect, and their intersection can be regarded as encoding the amplitudes $m(1234|2314)$.}
\label{ex1b}
\end{figure}

\paragraph{}
The biadjoint scalar tree amplitudes $m(\alpha|\beta)$ can be computed in a Feynman diagram expansion by summing over all tri-valent graphs which are planar with respect to both $\alpha$ and $\beta$. However, there are many other formulas for these amplitudes. A particularly novel formula for $m(\alpha|\beta)$ was given by Cachazo-He-Yuan, who expressed the amplitude as a residue pairing. \cite{chy14,dg14} In \cite{mizera1711}, Mizera pointed out that the CHY residue pairing is essentially the same thing as the intersection pairing of two cocycles in a certain cohomology theory.\footnote{The appropriate cohomology theory is cohomology with values in a local system. The underlying space is the open string moduli space $\mc{M}_{0,n}(\mb{R})$ and the local system is defined by the monodromies of the Koba-Nielsen factor. The cocycles in the pairing are represented in (twisted) de Rham theory by the Parke-Taylor factors associated to the two orderings of $\{1,...,n\}$. These are top-forms on the open string moduli space $\mc{M}_{0,n}(\mb{R})$.} An attractive formula for these intersection pairings was computed by Matsumoto in \cite{matsumoto98}. The work by AHBHY is more combinatorial. In \cite{ahbhy17}, the amplitudes $m(\alpha|\beta)$ are presented in terms of associahedra in kinematic space.\footnote{This realisation of the associahedron bears some resemblance to methods developed in the combinatorial literature, see \cite{csz15} for a review. The exact prescription given by AHBHY does, however, appear to be new.} One way to infer the amplitude from the associahedra is by studying `canonical forms' that have logarithmic singularities on the boundaries of the associahedra.\footnote{This fits into a general paradigm. The associahedra are defined by finitely many linear inequalities and equalities. For this reason, they are `linear semi-algebraic sets' or `positive geometries.' A volume form with logarithmic singularities on the boundary of a positive geometry is called a `canonical form' and a great deal can be said in general about these forms. See ref. \cite{ahbl17} for a recent elaboration of these ideas.} The biadjoint scalar tree amplitudes have aroused all of this interest because of their importance in understanding the relationship between tree-level gravity and Yang-Mills amplitudes. The partial amplitudes $m(\alpha|\beta)$ play an essential role in the `double copy' relation (equation \eqref{ftklt}) between gravity tree amplitudes and Yang-Mills partial tree amplitudes. \cite{bcj08} Indeed, Yang-Mills and gravity amplitudes can essentially be expanded as a sum over $m(\alpha|\beta)$ with certain numerators. Another source of interest is the ubiquitous presence of the biadjoint tree amplitudes in string theory. For instance, closed string tree amplitudes can be expanded in a basis of `$I$-integrals' $I(\alpha|\beta)$ whose leading term in the infinite tension limit ($\alpha' \rightarrow 0$) is $m(\alpha|\beta)$.\footnote{See equation (3.24) of \cite{st14} for the integral and equations (4.2), (4.3) and (4.4) for the expansion of the closed string amplitude.} \cite{st14} Likewise, open string tree amplitudes can be expanded in a basis of `$Z$-integrals' $Z_{\alpha}(\beta)$ whose leading term in the infinite tension limit is $m(\alpha|\beta)$.\footnote{See equation (2.2) of \cite{ms17} for the integral and (2.17) for the expansion of the open string amplitude.} \cite{ms17} The open and closed string tree amplitudes are related by the Kawai-Lewellen-Tai relations. \cite{klt} Following Mizera \cite{mizera1706}, we denote the inverse KLT kernel by $m_{\alpha'}(\alpha|\beta)$ since, in the infinite tension limit, the leading term in $m_{\alpha'}(\alpha|\beta)$ is given (up to a possible factor) by $m(\alpha|\beta)$. This means that we might regard the inverse KLT kernel $m_{\alpha'}(\alpha|\beta)$ as a the amplitude $m(\alpha|\beta)$ with $\alpha'$ corrections. Perhaps surprisingly, $m_{\alpha'}(\alpha|\beta)$ has a tidy description in terms of the open string moduli space $\mc{M}_{0,n}(\mb{R})$ and the monodromy factors picked up by the Koba-Nielsen factor when points collide. The compactified moduli space $\overline{\mc{M}_{0,n}}(\mb{R})$ is tesselated by $(n-1)!/2$ copies of the associahedron on $n-1$ letters. (See \cite{devadoss}, theorem 3.1.3, for instance.) Mizera proved that the functions $m_{\alpha'}(\alpha|\beta)$ can be obtained as the intersections of these associahedra regarded as twisted cycles in $\mc{M}_{0,n}(\mb{R})$. \cite{mizera1708} Some of the resulting formulas for $m_{\alpha'}(\alpha|\beta)$ were encountered earlier in the mathematical literature on twisted de Rham theory (see \cite{ky94}, in particular). Mizera's result gives a new interpretation of the KLT relation as a purely cohomological statement, namely it is the analogue in twisted de Rham theory (or cohomology with values in a local system) of Riemann's period relations. This interpretation is presented and discussed in \cite{mizera1708}.

\paragraph{}
Our main result is a new formula for the amplitudes $m(\alpha|\beta)$. All $n$-point partial amplitudes are described by a single object, called a `fan,' in dual kinematic space $\mc{K}_n^*$. Individual partial amplitudes $m(\alpha|\beta)$ are given by the `volumes' of the intersections of dual associahedra in this fan. This result is stated precisely in section \ref{statement} as theorem \ref{theorem} and corollary \ref{valley}. Once the necessary definitions are given, the proof of this result is essentially tautological. After all, associahedra and dual associahedra are just a fancy way of encoding the Feynman diagram sum of trivalent graphs. Given this, why do we need another formula for $m(\alpha|\beta)$? Especially since so many formulas already exist? We think our formula has two intrinsic merits. First, it is a compact and symmetric way to present the relevant combinatorics. Second, it leads us to study a certain object (the fan) which encodes all of the partial amplitudes at once. In fact, as we point out at the end of section \ref{comments}, our main result, restated in the language of toric geometry, is that all the $n$-point partial amplitudes $m(\alpha|\beta)$ are described by a single toric variety. Particular partial amplitudes $m(\alpha|\beta)$ correspond to toric subvarieties (or, really, cycles in the homology). This appears to us as a novel description of the amplitudes. We can only hope that these new geometric presentations of the amplitudes will suggest new approaches to the Yang-Mills and gravity amplitudes and to the double copy relations between them. We offer speculations to this end during the discussion in section \ref{comments}. We also relegate several technical (but quite interesting) loose ends to this discussion. Before proving our results in section \ref{statement}, our key `canonical embedding' construction is described in \ref{embed}. We explain all the essential definitions in section \ref{review}, where we also review the relevant results that we need from AHBHY (ref. \cite{ahbhy17}). Our second result is an observation concerning the inverse KLT kernel. The diagonal components of the inverse KLT kernel|that is, $m_{\alpha'}(\alpha|\alpha)$|are also related to the fan in dual kinematic space which encodes the partial amplitudes $m(\alpha|\beta)$. To obtain the components $m_{\alpha'}(\alpha|\alpha)$ we have to introduce a lattice on kinematic space. This lattice amounts to imposing that the Mandelstam variables $s_{ij}$ are integer multiples of $1/\alpha'$. (It seems natural to introduce a lattice here since, as we have just remarked, toric geometry is implicit throughout this paper and the construction of toric varieties is predicated on lattices.) Having introduced the lattice, we observe in section \ref{stringyformula} that the components $m_{\alpha'}(\alpha|\alpha)$ can be given as a certain sum over lattice points. This curious, and potentially superficial observation, sparks a number of further speculations during the discussion in section \ref{comments}. The main text of the paper is largely self contained. However, for the reader's convenience, several definitions and theorems mentioned briefly in the text are reviewed at greater length in appendix \ref{old}, where we also give references to the many expositions that already exist.

\paragraph{Acknowledgements.}
It is a pleasure to thank Lionel Mason and Eduardo Casali for our many engaging discussions and for their helpful comments on the draft. I acknowledge financial support from the Rhodes trust, Merton college, John Moussouris, and Universities New Zealand. Finally, I am indebted to the Jones family for their hospitality.

\section{Review and definitions}
\label{review}
\begin{figure}
\begin{center}
\includegraphics[scale=0.40]{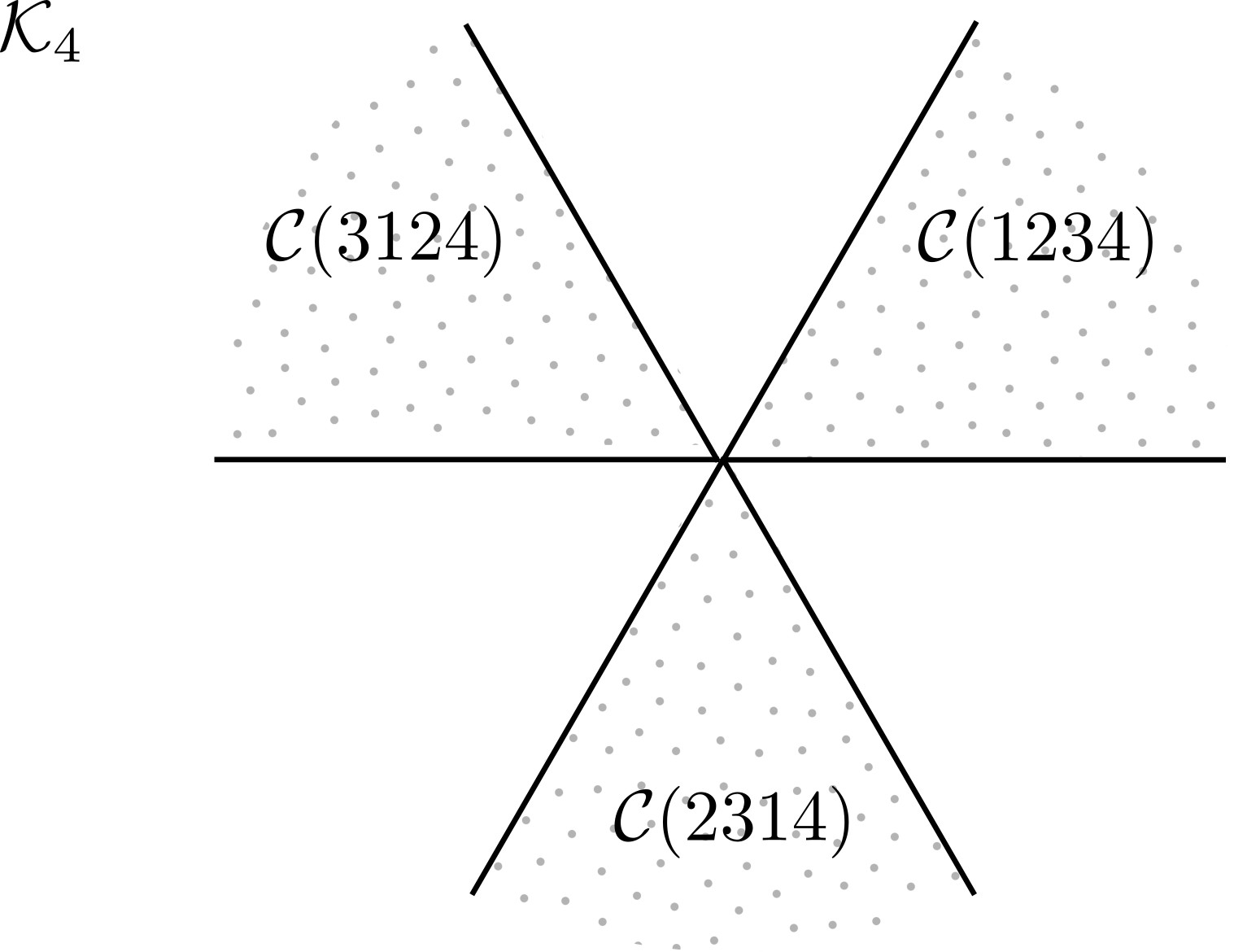}
\end{center}
\caption{The three cones $\mc{C}(\alpha)$ in the kinematic space $\mc{K}_4$ for four points.}
\label{ex4a}
\end{figure}
In this section we describe AHBHY's construction of associahedra in `kinematic space.' At $n$ points, their construction gives $(n-1)!/2$ distinct associahedra. This is reminiscent of the open string moduli space, $\mc{M}_{0,n}(\mb{R})$, which (after blow-ups) is tiled by $(n-1)!/2$ associahedra. Unlike the tiles of $\mc{M}_{0,n}(\mb{R})$, the associahedra that AHBHY construct in $\mc{K}_n$ do not intersect. However, we will see in section \ref{embed} that it is natural to regard the dual associahedra as intersecting in dual kinematic space, $\mc{K}_n^*$.

\subsection{Cones in kinematic space}
\label{reviewcones}
To begin, let us define `kinematic space' at $n$ points, $\mc{K}_n$, to be the space of all Mandelstam variables subject to momentum conservation.\footnote{See also section 2 of \cite{ahbhy17}.} We can present $\mc{K}_n$ explicitly as follows. Let $\{s_{ij}\}$, where $1\leq i,j \leq n$, be the $n(n-1)/2$ Mandelstam variables. We impose the $n$ momentum conservation relations
$$
\sum_{j=1}^ns_{ij} = 0.
$$
We can use these relations to remove one index, say $n$, from our formulas. That is, we consider only those $s_{ij}$ with $1\leq i,j \leq n-1$. There are $(n-1)(n-2)/2$ such variables and we regard them as coordinates on a vector space
$$
\mc{V}_n = \mb{R}^{\frac{n(n-3)}{2}+1}.
$$
The remaining momentum conservation relation on these variables is
$$
\sum_{i,j=1}^{n-1} s_{ij} = 0,
$$
which presents the kinematic space $\mc{K}_n$ as a hyperplane in $\mc{V}_n$. For example, at four points, the vector space $\mc{V}_4$ is $\mb{R}^3$ with coordinates $(s_{12},s_{23},s_{13})$ and the kinematic space $\mc{K}_4$ is the hyperplane
$$
s_{12}+s_{23}+s_{13} = 0.
$$
Now let $\alpha$ be an ordering of $\{1,...,n\}$ such that $\alpha(n) = n$. Associated to each such ordering, AHBHY associate a cone in $\mc{K}_n$.\footnote{This is done in section 3.2 of \cite{ahbhy17}.} The cone associated to $\alpha$ is cut out by the inequalities
\begin{equation}
\label{inequalities}
X_{S} = - \frac{1}{2}\sum_{i,j\in S} s_{ij} \geq 0
\end{equation}
for all subsets $S\subset \{1,...,n-1\}$ which are `consecutive' with respect to $\alpha$.\footnote{By `consecutive' I mean that $S$ can be written as $\{\alpha(a),\alpha(a+1),...,\alpha(a+k)\}$ for some positive integers $a$ and $k$. For instance $\{1,2,3\}$ is consecutive with respect to $\alpha = (1324)$, but $\{1,2\}$ is not.} There are
$$
{n\choose 2} - (n-1)= \frac{n(n-3)}{2} + 1
$$
such subsets $S$, corresponding to the diagonals of the $n$-gon. However, $X_{123...n-1}$ is identically zero on the hyperplane $\mc{K}_n$. So only
$$
\frac{n(n-3)}{2}
$$
inequalities are implied by equation \eqref{inequalities}. We write
$$
\mc{C}(\alpha) = \left\{(s_{ij})\in\mc{K}_n\,|\,\text{satisfying the inequalities, equation}~\eqref{inequalities} \right\}
$$
for this cone. The cone $\mc{C}(\alpha)$ is `polyhedral' in the sense that it has finitely many flat sides.\footnote{This is in contrast to a cone generated by a sphere, or any arbitrary convex set. See appendix \ref{oldcones} for a review of definitions and results concerning polyhedral cones.} Moreover, the inequalities that define $\mc{C}(\alpha)$ are linearly independent, which means that $\mc{C}(\alpha)$ contains an interior region of full dimension, $\dim \mc{K}_n = n(n-3)/2$. There are $(n-1)!$ permutations $\alpha$ such that $\alpha(n) = n$. However, the inequalities, equation \eqref{inequalities}, do not define $(n-1)!$ distinct cones $\mc{C}(\alpha)$. Given a permutation $\alpha$ such that $\alpha(n) = n$, define the reversed permutation $\bar\alpha$ as
$$
(\bar\alpha(1),\bar\alpha(2),...,\bar\alpha(n-1),\bar\alpha(n)) = (\alpha(n-1),\alpha(n-2),...,\alpha(1),\alpha(n)).
$$
The inequalities, equation \eqref{inequalities}, are identical for $\alpha$ and its reverse $\bar\alpha$. That is, $\mc{C}(\alpha) = \mc{C}(\bar\alpha)$. For this reason, we obtain only $(n-1)!/2$ distinct cones. In the four point example, we obtain $3!/2 =3$ cones. These are
\begin{align*}
\mc{C}(1234) = \mc{C}(3214) = \{(s_{ij})\,|\,-s_{12},-s_{23}\geq 0\},\\
\mc{C}(3124) = \mc{C}(2134) = \{(s_{ij})\,|\,-s_{12},-s_{13}\geq 0\},\\
\mc{C}(2314) = \mc{C}(1324) = \{(s_{ij})\,|\,-s_{13},-s_{23}\geq 0\}.
\end{align*}
We show these cones in figure \ref{ex4a}. To do this, we employ a non-orthogonal basis
\begin{equation}
\label{hexa}
e_{12} = \begin{bmatrix} 0\\1\end{bmatrix},\qquad\text{and}\qquad e_{23} = \begin{bmatrix} \sqrt{3}/2\\-1/2\end{bmatrix}
\end{equation}
for the plane, so that the point $(-s_{12},-s_{23})$ in $\mc{K}_4$ is represented in the plane by $-s_{12}e_{12} - s_{23}e_{23}$. We choose this basis to make the symmetries of the construction more apparent in our drawings. Before we introduce associahedra in the following subsection, let us briefly comment on the physical significance of what we have done so far. Recall that the Mandelstam variables are $s_{ij} = 2k_i\cdot k_j$, where $k_i$ and $k_j$ are massless momentum vectors. The sign of $s_{ij}$ then encodes whether $k_i$ and $k_j$ are time directed in the same way (i.e. both forward or both past directed) or in opposite ways. To see this, write $k_i = \omega_i(\epsilon_i,n_i)$ for a sign $\epsilon_i = \pm1$, a positive energy $\omega_i$, and a unit vector $n_i$. In Lorentzian signature we have
$$
k_i\cdot k_j = \omega_i\omega_j (\epsilon_i\epsilon_j - \cos \theta),
$$
where $\theta$ is the angle between $n_i$ and $n_j$. So $s_{ij} \geq 0$ iff $\epsilon_i\epsilon_j = 1$ and $s_{ij}\leq 0$ iff $\epsilon_i\epsilon_j = -1$. For example, at four points, both $s_{12}$ and $s_{23}$ are negative inside the cone $\mc{C}(1234)$. Momentum conservation implies that $s_{13}$ is positive inside the cone. This means that $k_1$ and $k_3$ are both future-pointing (or both past-pointing) while $k_2$ is past-pointing (resp. future-pointing). So we see that some of the inequalities in equation \eqref{inequalities} have a physical interpretation in terms of positive and negative energies.

\subsection{Associahedra in kinematic space}
\begin{figure}
\begin{center}
\includegraphics[scale=0.40]{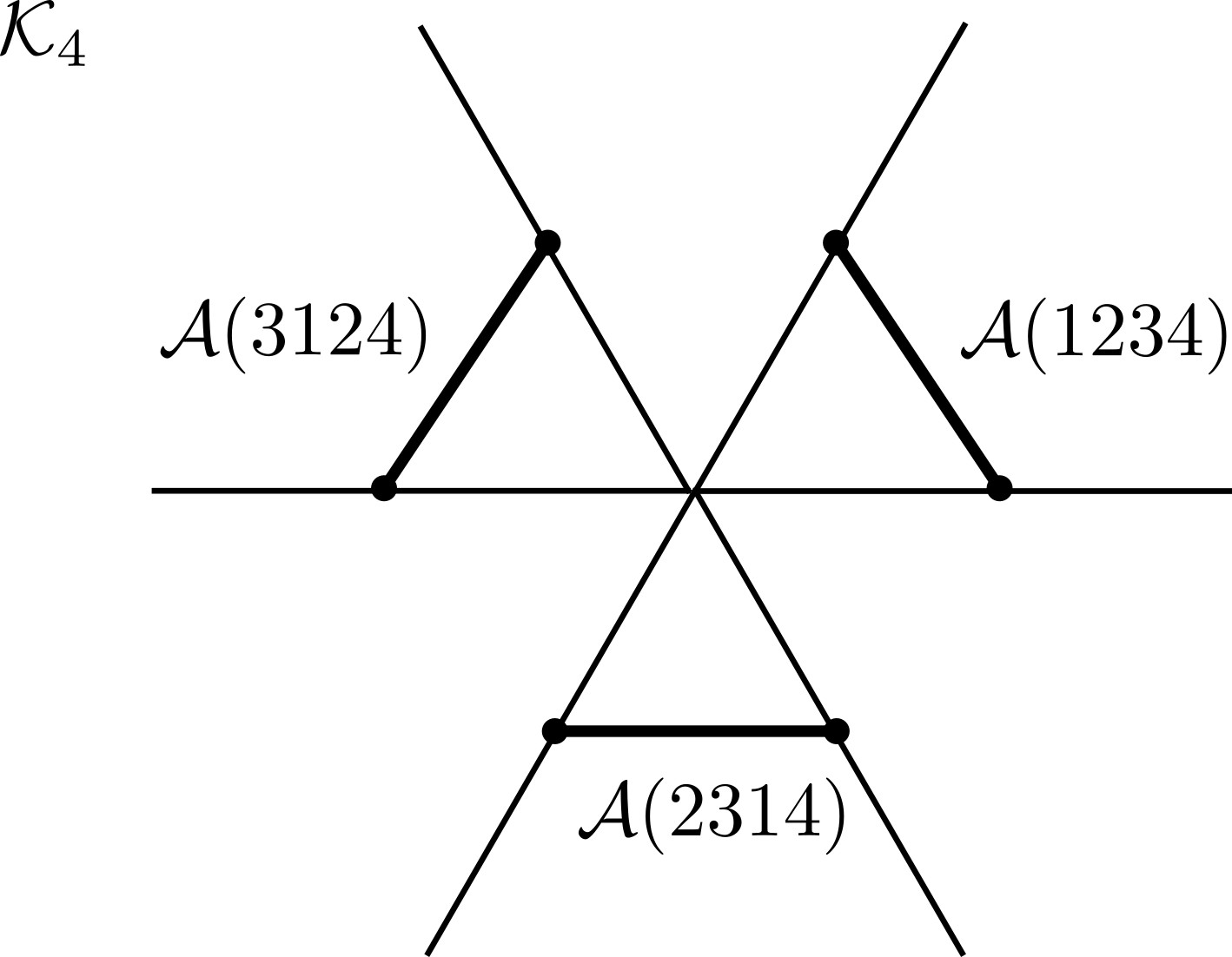}
\end{center}
\caption{The three realisations $\mc{A}(\alpha)$ of the associahedron in the kinematic space $\mc{K}_4$ for four points.}
\label{ex4b}
\end{figure}
AHBHY construct $(n-1)!/2$ associahedra in kinematic space $\mc{K}_n$ by intersecting each of the cones $\mc{C}(\alpha)$ with an appropriate hyperplane.\footnote{See again section 3.2 of \cite{ahbhy17}.} For each ordering $\alpha$, with $\alpha(n)=n$, define the hyperplane
$$
H(\alpha) = \left\{s_{\alpha(i),\alpha(j)} = \text{constant}\,|\,\text{for all non-adjacent pairs $i,j$ less than $n$}\right\}.
$$
There are $(n-2)(n-3)/2$ pairs of non-adjacent $(i,j)$ for $i$ and $j$ less than $n$. So the dimension of $H(\alpha)$ is
$$
\frac{n(n-3)}{2} - \frac{(n-2)(n-3)}{2} = n-3.
$$
The intersection
$$
\mc{A}(\alpha) = \mc{C}(\alpha)\cap H(\alpha) \subset \mc{K}_n
$$
defines a polytope. For generic choices of the constants that define $H(\alpha)$, $\mc{A}(\alpha)$ is a realisation of the associahedron. The polytope $\mc{A}(\alpha)$ lies in the $n-3$ dimensional plane $H(\alpha)$. Nevertheless, the facets of $\mc{A}(\alpha)$ (these are the top dimension faces of $\mc{A}(\alpha)$) generate the cone $\mc{C}(\alpha)$, which bounds an interior region of dimension $n(n-3)/2$. Since all faces of $\mc{A}(\alpha)$ can be obtained as the intersections of facets, it follows that all faces of $\mc{A}(\alpha)$ are contained in the boundary of the cone $\mc{C}(\alpha)$. For example, at four points, consider the cone
$$
\mc{C}(1234) = \{(s_{ij})\,|\,-s_{12},-s_{23}\geq 0\}.
$$
We choose the hyperplane
$$
H(1234) = \{(s_{ij})\,|\,s_{13} = \sqrt{3}/2\} .
$$
Then, using the basis introduced in equation \eqref{hexa}, the associahedron is given by the convex hull
$$
\mc{A}(1234) = \Conv \left(\begin{bmatrix}1 \\ 0\end{bmatrix}, \begin{bmatrix} 1/2\\ \sqrt{3}/2\end{bmatrix}\right).
$$
This is shown in figure \ref{ex4b}, together with the associahedra $\mc{A}(3124)$ and $\mc{A}(2314)$.

\subsection{Dual cones and associahedra}
\label{reviewdual}
\begin{figure}
\begin{center}
\includegraphics[scale=0.40]{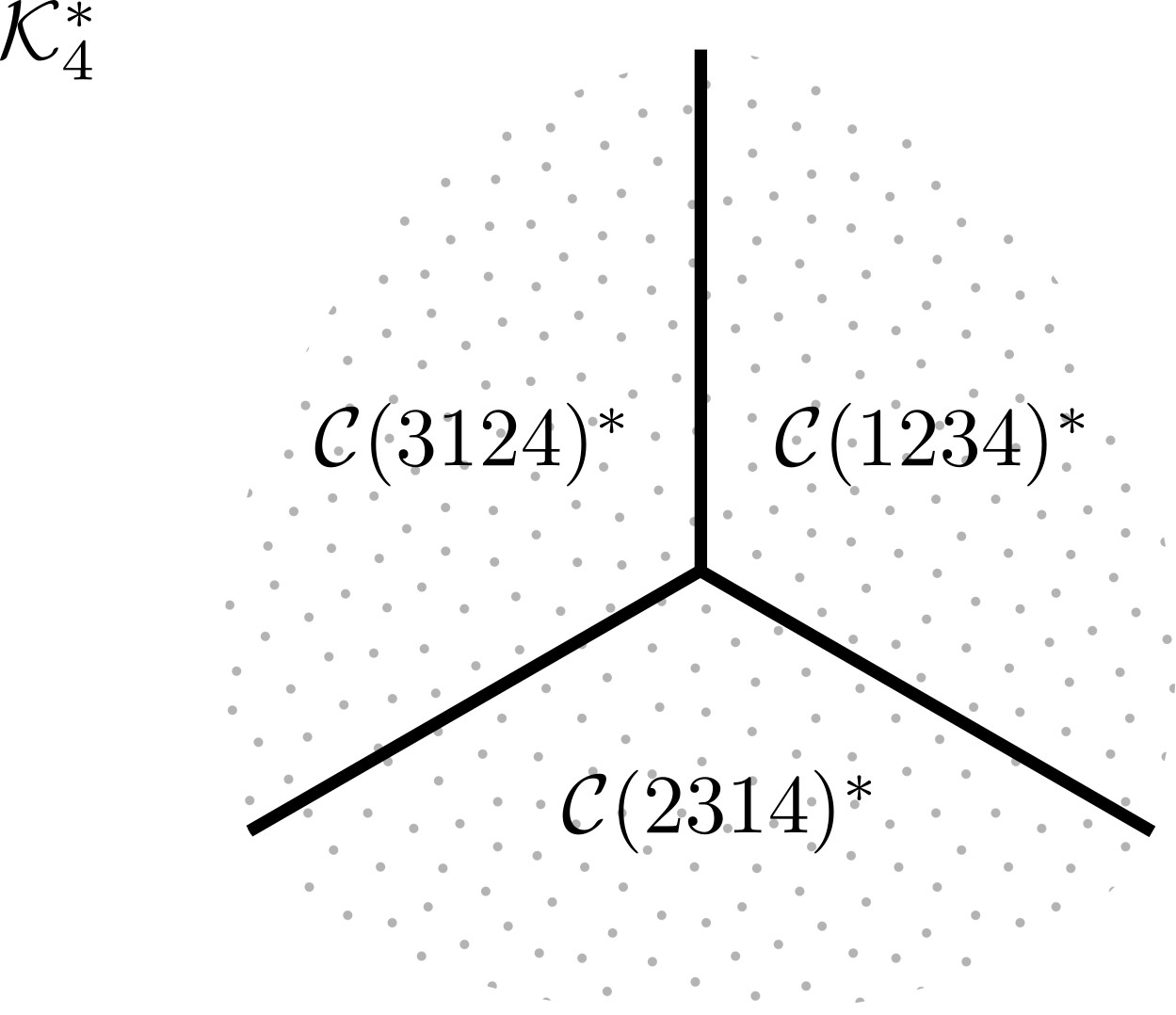}
\end{center}
\caption{The three dual cones $\mc{C}(\alpha)$ in $\mc{K}_4^*$.}
\label{ex4c}
\end{figure}
Let $\mc{K}_n^*$ be the dual (as a vector space) of $\mc{K}_n$. There is a standard notion of a `dual cone.' Define
\begin{equation}
\label{dualconedef}
\mc{C}(\alpha)^* = \left\{W \in \mc{K}_n^*\,|\,W\cdot Z \geq 0 ~\text{for all}~Z \in \mc{C}(\alpha) \right\}.
\end{equation}
The inequalities that define $\mc{C}(\alpha)$, equation \eqref{inequalities}, can be re-expressed as
$$
W_{S}\cdot Z \geq 0,
$$
where the covectors $W_{S}\in \mc{K}_n^*$ are defined by
\begin{equation}
\label{gendef}
W_{S}\cdot Z = X_{S}.
\end{equation}
Given this, the cone $\mc{C}(\alpha)^*$ is given by the `conic hull' or `positive span'\footnote{See equation \eqref{defconichull} in appendix \ref{oldcones} for conic hulls.}
$$
\mc{C}(\alpha)^* = \Cone \left\{ W_{S}\,|\,\text{for}~S~\text{consecutive with respect to}~\alpha \right\}.
$$
In the four point example we have, for instance,
$$
\mc{C}(1234)^* = \Cone (W_{12},W_{23}).
$$
In fact, the three dual cones $\mc{C}(1234)^*$, $\mc{C}(3124)^*$, and $\mc{C}(2314)^*$ fill the dual space $\mc{K}_4^*$, as shown in figure \ref{ex4c}.\footnote{This is a consequence of the fact that the three cones $\mc{C}(1234)$, $\mc{C}(3124)$, and $\mc{C}(2314)$ in $\mc{K}_4$ are the tangent cones of an isosceles triangle (translated so that their apexes are all at the origin). In general, the sum of dual tangent cones of a convex polytope fills the dual vector space. See theorem \ref{dualcones}.} What about dual associahedra? Unlike the cones, there is no natural notion of a dual associahedron in $\mc{K}_n^*$. Instead, we must restrict to the $(n-3)$-hyperplane $H(\alpha)$ that contains $\mc{A}(\alpha)$. Then the `dual polytope' is given by
$$
\mc{A}(\alpha)^* = \{Y \in H(\alpha)^* \,|\, Y\cdot Z \geq -1~\text{for all}~Z \in \mc{A}(\alpha) \} \subset H(\alpha)^*.
$$
The duality operation swaps dimension-$k$ faces in $\mc{A}(\alpha)$ for codimenion-$(k+1)$ faces in $\mc{A}(\alpha)^*$. Thus, vertices in $\mc{A}(\alpha)$ become facets in $\mc{A}(\alpha)^*$, and so on. To give a concrete example, let's return to the four point case and consider the line $H(1234)$. Using $-s_{12}$ as a coordinate on this line, the associahedron is the interval
$$
\mc{A}(1234) = \left[0,\sqrt{3}/2\right] \subset H(1234).
$$
The dual polytope is then
$$
\mc{A}(1234)^* = \left[-2/\sqrt{3},\infty\right) \subset H(1234)^*.
$$
If we compactify at infinity, we can regard $H(1234)$ as $\mb{RP}^1$. Then the duality operation has sent a 1-simplex $\mc{A}(1234)$ to a dual 1-simplex $\mc{A}(1234)^*$. We illustrate this process in figure \ref{exb1} so as to emphasise that the dual associahedron defined this way is not yet embedded into dual kinematic space. Such an embedding is an additional construction which inovlves making choices. This is analogous to the choices that are involved in AHBHY's embeddings of the associahedron into kinematic space. We choose a particular embedding in the next section.
\begin{figure}
\begin{center}
\includegraphics[scale=0.40]{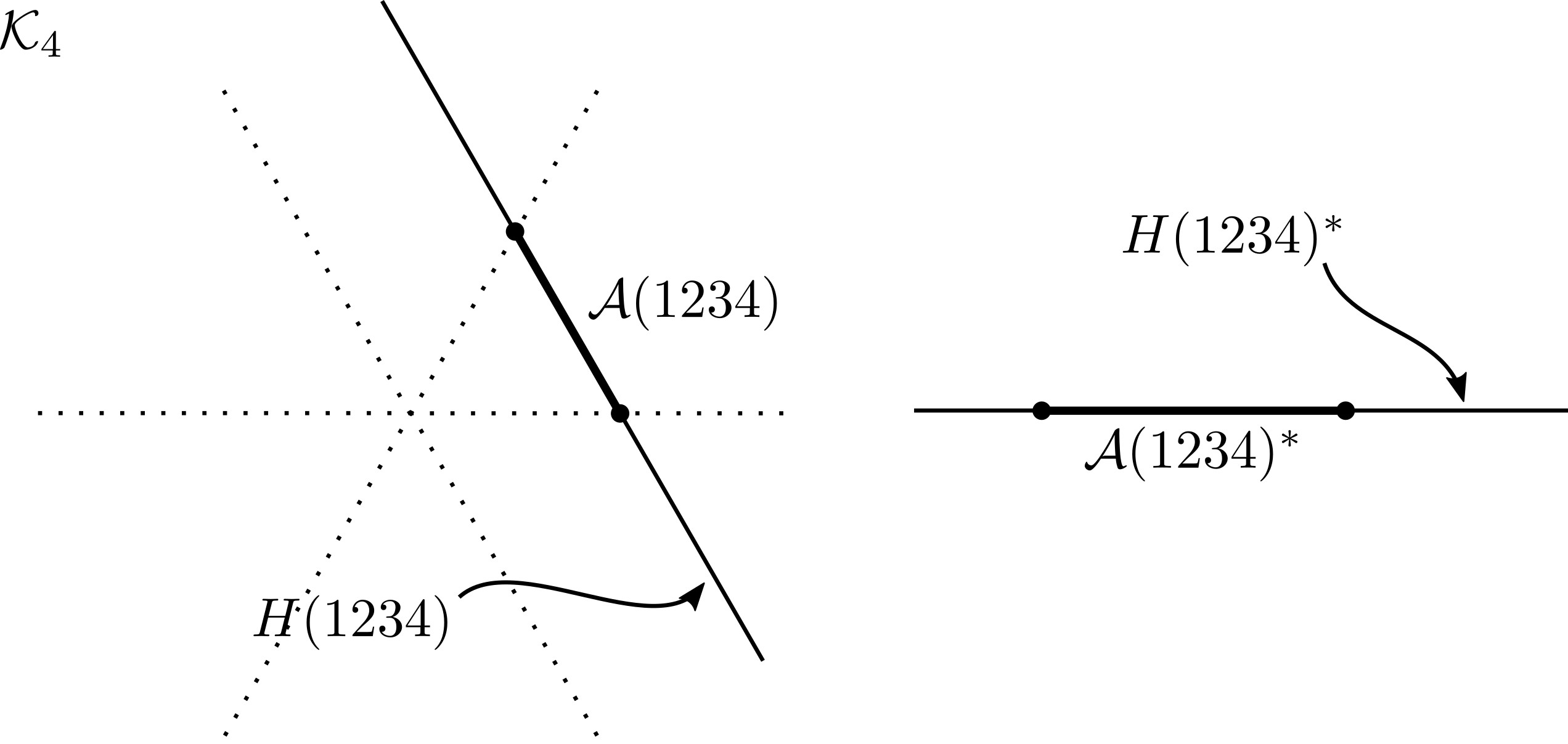}
\end{center}
\caption{The dual polytope construction for a 1-simplex associahedron encountered in the four point example. Notice that duality for a polytope $P$ in $\mc{K}_n$ is not defined with respect to $\mc{K}_n$, but rather with respect to the smallest linear subspace containing $P$. This means that the dual associahedron $\mc{A}(1234)^*$ is not naturally regarded as living in $\mc{K}_n^*$ until an embedding is chosen. We make a choice of embedding in section \ref{embed}.}
\label{exb1}
\end{figure}

\section{Embedding dual associahedra}
\label{embed}
AHBHY construct $(n-1)!/2$ embeddings of the associahedron, which has dimension $n-3$, into a higher dimensional vector space, $\mc{K}_n$. We will now construct $(n-1)!/2$ embeddings of the dual associahedra $\mc{A}(\alpha)^*$ (or, in fact, their faces) into dual kinematic space $\mc{K}_n^*$. The idea is to embed the (faces of the) dual associahedron $\mc{A}(\alpha)^*$ into the (boundary of the) dual cone $\mc{C}(\alpha)^*$. This can be done in a canonical way such that the embedded associahedra $\mc{A}(\alpha)^* \subset \mc{K}_n^*$ have combinatorially meaningful intersections. In defining our embedding we will make use of the vectors
$$
W_S \in \mc{K}_n^*,\qquad S \subset \{1,...,n-1\},
$$
that we defined in equation \eqref{gendef}. These vectors can be used as generators for the cones $\mc{C}(\alpha)$. Now recall that there is a 1-1 correspondence between codimension $k$ faces of $\mc{A}(\alpha) \subset H(\alpha)$ and dimension $k-1$ faces of $\mc{A}(\alpha)^*\subset H(\alpha)^*$. In particular, the facets (or codimension-1 faces) of $\mc{A}(\alpha) \subset H(\alpha)$ correspond to the vertices of $\mc{A}(\alpha)^*\subset H(\alpha)^*$. A facet of $\mc{A}(\alpha) \subset H(\alpha)$ is defined by a single inequality,
$$
W_S\cdot Z \geq 0,
$$
for some subset $S$. Let $Y_S \in H(\alpha)^*$ be the vertex of $\mc{A}(\alpha)^*\subset H(\alpha)^*$ corresponding to the facet $W_S\cdot Z = 0$ of $\mc{A}(\alpha)$. Then any dimension-$k$ face of $\mc{A}(\alpha)^*\subset H(\alpha)^*$ is given by the convex hull
$$
\Conv (Y_{S_1},...,Y_{S_k}) \subset H(\alpha)^*,
$$
for some subsets $S_i$ labelling the vertices of $\mc{A}(\alpha)^*$. We map this face to the convex hull
$$
\Conv  (W_{S_1},...,W_{S_k}) \subset  \mc{C}(\alpha)^*
$$
in $\mc{K}_n^*$. In particular, we map the vertex $Y_{S}$ to the vector $W_{S}$. In this way, we can embed all the faces of $\mc{A}(\alpha)^*$ into $\mc{K}_n^*$. In other words, we have an embedding of the boundary $\p\mc{A}(\alpha)^*$ into $\mc{K}_n^*$. We call this the `canonical embedding' of the faces and will typically abuse notation by denoting the embedded faces as $\p\mc{A}(\alpha)^*$. In general, the canonically embedded faces $\p\mc{A}(\alpha)^*$ do not bound a dimension $n-3$ polytope. (This does happen for four points, but not at higher points.) One can obtain an embedding of all $\mc{A}(\alpha)^*$ into $\mc{K}_n^*$ by choosing a triangulation of $\mc{A}(\alpha)^*$ into $(n-3)$-simplices. This idea is illustrated in figure \ref{exb2}.
\begin{figure}
\begin{center}
\includegraphics[scale=0.40]{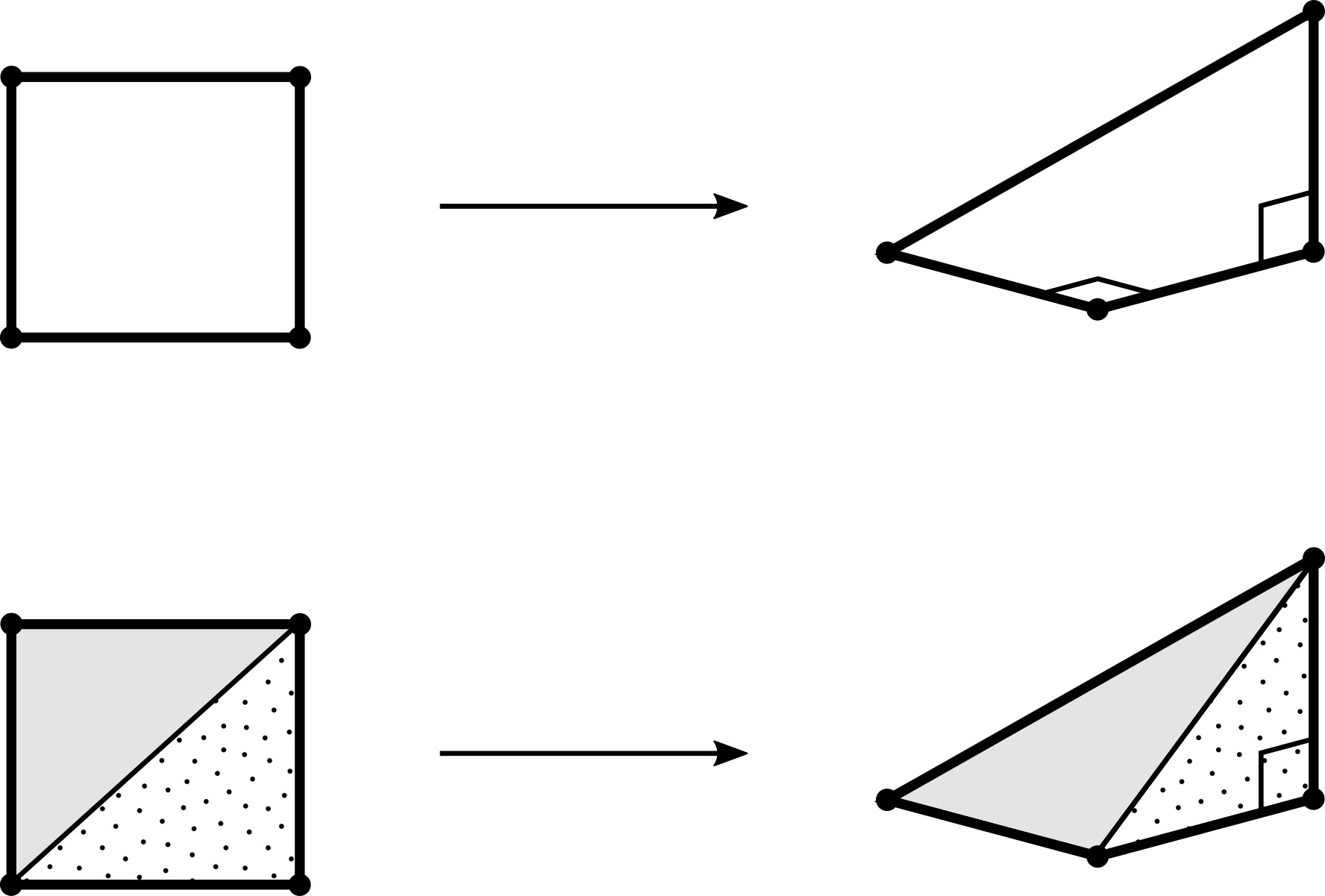}
\end{center}
\caption{The faces of a polytope (like the square) can be embedded in a higher dimensional space in such a way that they do not bound a polytope. The original polytope can still be embedded in the higher dimensional space by choosing a triangulation.}
\label{exb2}
\end{figure}
However, we will not do this. The reason is that all the combinatorial data that we need for amplitudes is already contained in the faces of the dual associahedra. (Recall that these faces are in direct correspondence with vertices of the associahedron, which is to say, with trivalent graphs.) Finally, let us denote the union of all the canonically embedded faces by
$$
FN_n = \bigcup_{\alpha} \p\mc{A}(\alpha)^* \subset \mc{K}_n^*
$$
(where we write `FN' for `face net'). We now give two examples. Consider first the four point example. The three dual vectors are
$$
W_{12}, W_{13}, W_{23},
$$
and the dual cones are given by
\begin{align*}
\mc{C}(1234)^* = \Cone( W_{12},W_{23} ),\\
\mc{C}(2314)^* = \Cone( W_{13},W_{23} ),\\
\mc{C}(3124)^* = \Cone( W_{12},W_{13} ).
\end{align*}
Then the canonically embedded dual associahedra are
\begin{align}
\label{ex4da}
\mc{A}(1234)^* = \Conv( W_{12},W_{23} ),\\
\mc{A}(2314)^* = \Conv( W_{13},W_{23} ),\\
\label{ex4da2}
\mc{A}(3124)^* = \Conv( W_{12},W_{13} ).
\end{align}
That is, the dual associahedra tile an isosceles triangle in $\mc{K}_4^*$. See figure \ref{ex4d}. So $FN_4$ gives the vertices of a triangle.
\begin{figure}
\begin{center}
\includegraphics[scale=0.40]{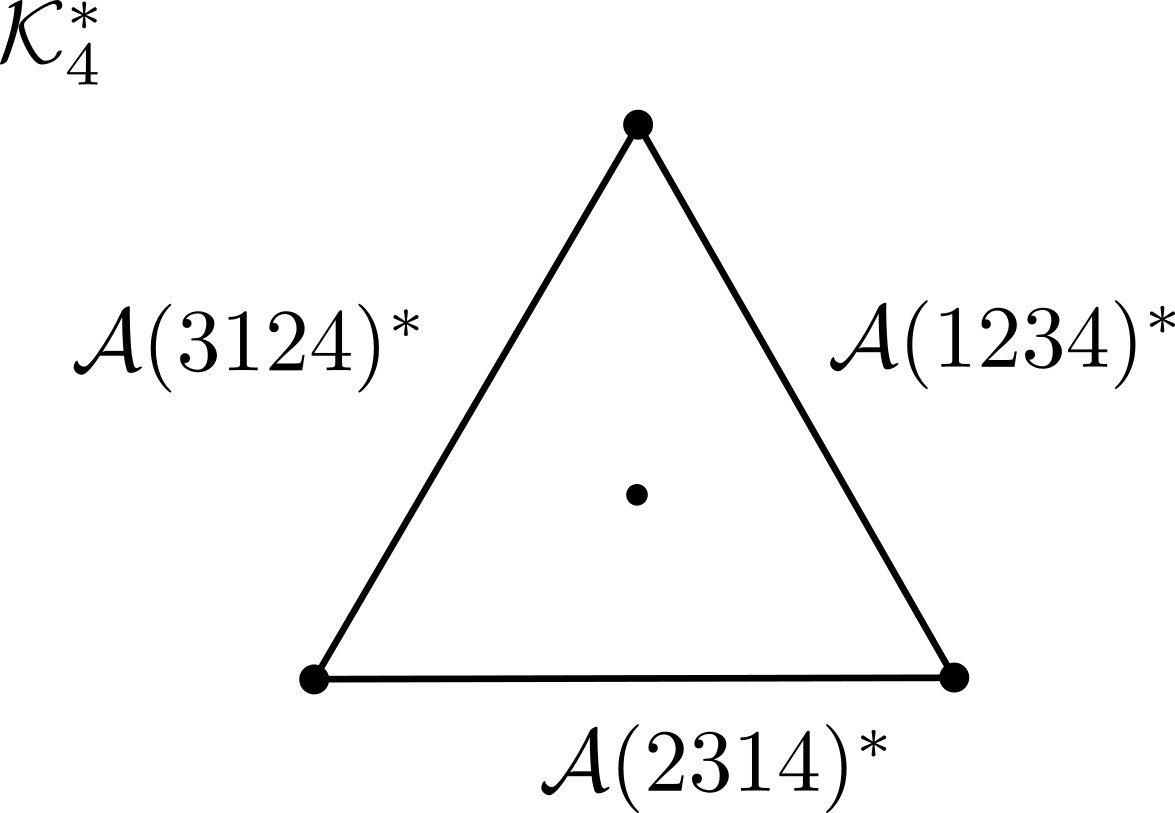}
\end{center}
\caption{The three dual associahedra $\mc{A}(\alpha)^*$ canonically embedded in $\mc{K}_4^*$ form a triangle.}
\label{ex4d}
\end{figure}
The situation at five points is less straightforward. In this case, there are $(5-1)!/2 = 12$ associahedra, each being a pentagon. The dual of a pentagon is a pentagon. So $FN_5$ will be formed by embedding the boundaries of 12 pentagons into $\mc{K}_5^*$. The vertices of the pentagons will each be mapped under the canonical embedding to one of the following ten vectors,
$$
\{W_{12},W_{13},W_{14},W_{23},W_{24},W_{34},W_{234},W_{134},W_{124},W_{123}\}.
$$
Though we cannot easily sketch the result of this embedding, we can draw the `net' associated to it, which is shown in figure \ref{ex5net}. Vertices in the net with the same label are identified in $\mc{K}_5^*$ under the embedding. Figure \ref{ex5net} can be arrived at from the analogous diagram showing $\overline{\mc{M}_{0,5}}(\mb{R})$ tiled by pentagons. This procedure is described in figure \ref{exb3}. Examining figure \ref{ex5net}, we see that each vertex is contained in the boundary of 6 distinct faces. Moreover, there are 12 faces altogether and 10 vertices. This suggests that the embedded associahedra in $\mc{K}_5^*$ tile a `halved' or `degenerate' dodecahedron, since a dodecahedron has 20 vertices, each of which meets 3 of its 12 faces.
\begin{figure}
\begin{center}
\includegraphics[scale=0.40]{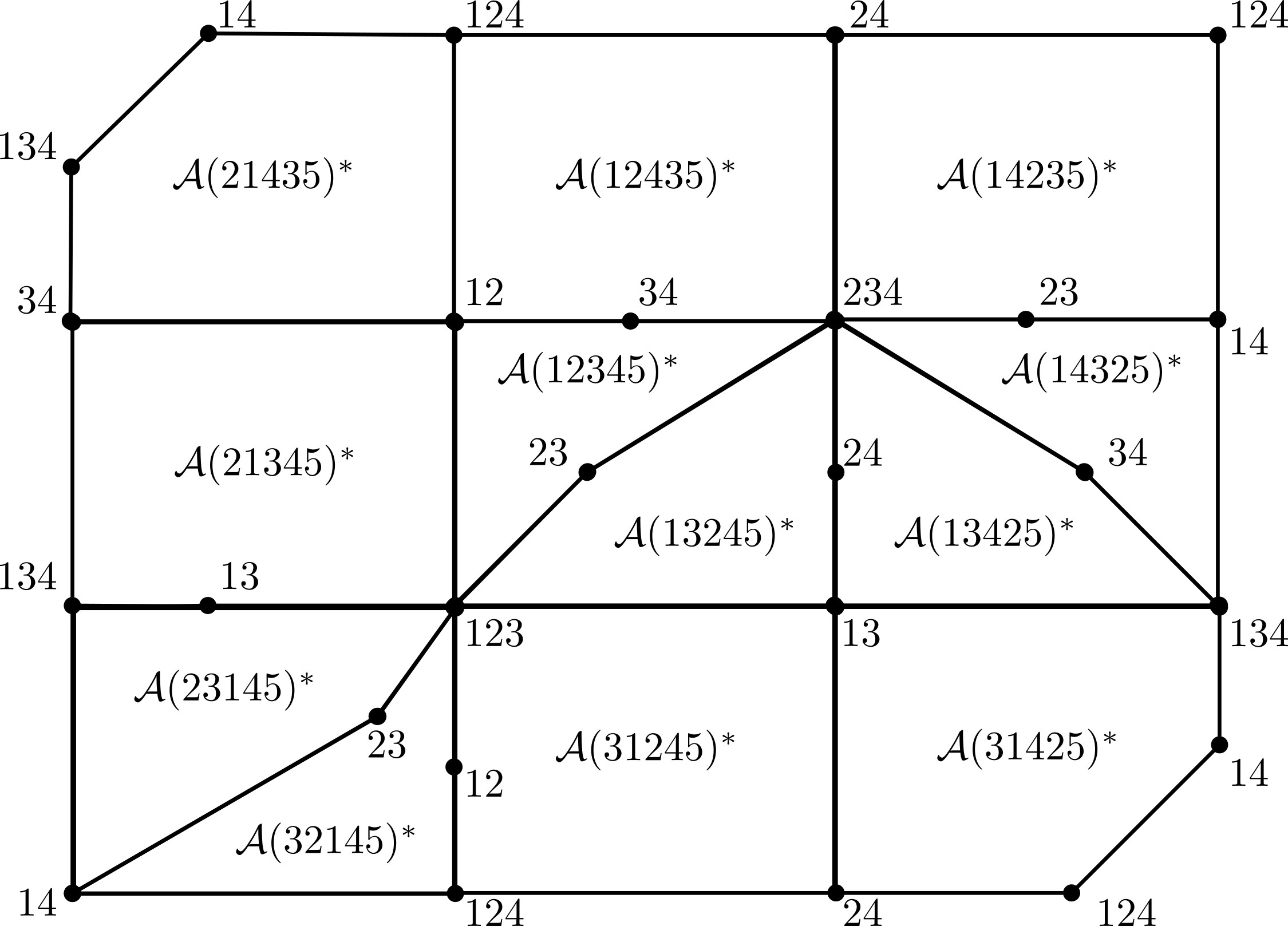}
\end{center}
\caption{The net which corresponds to the embedded associahedra in $\mc{K}_5^*$. Vertices with the same label are identified in the embedding.}
\label{ex5net}
\end{figure}
\begin{figure}
\begin{center}
\includegraphics[scale=0.40]{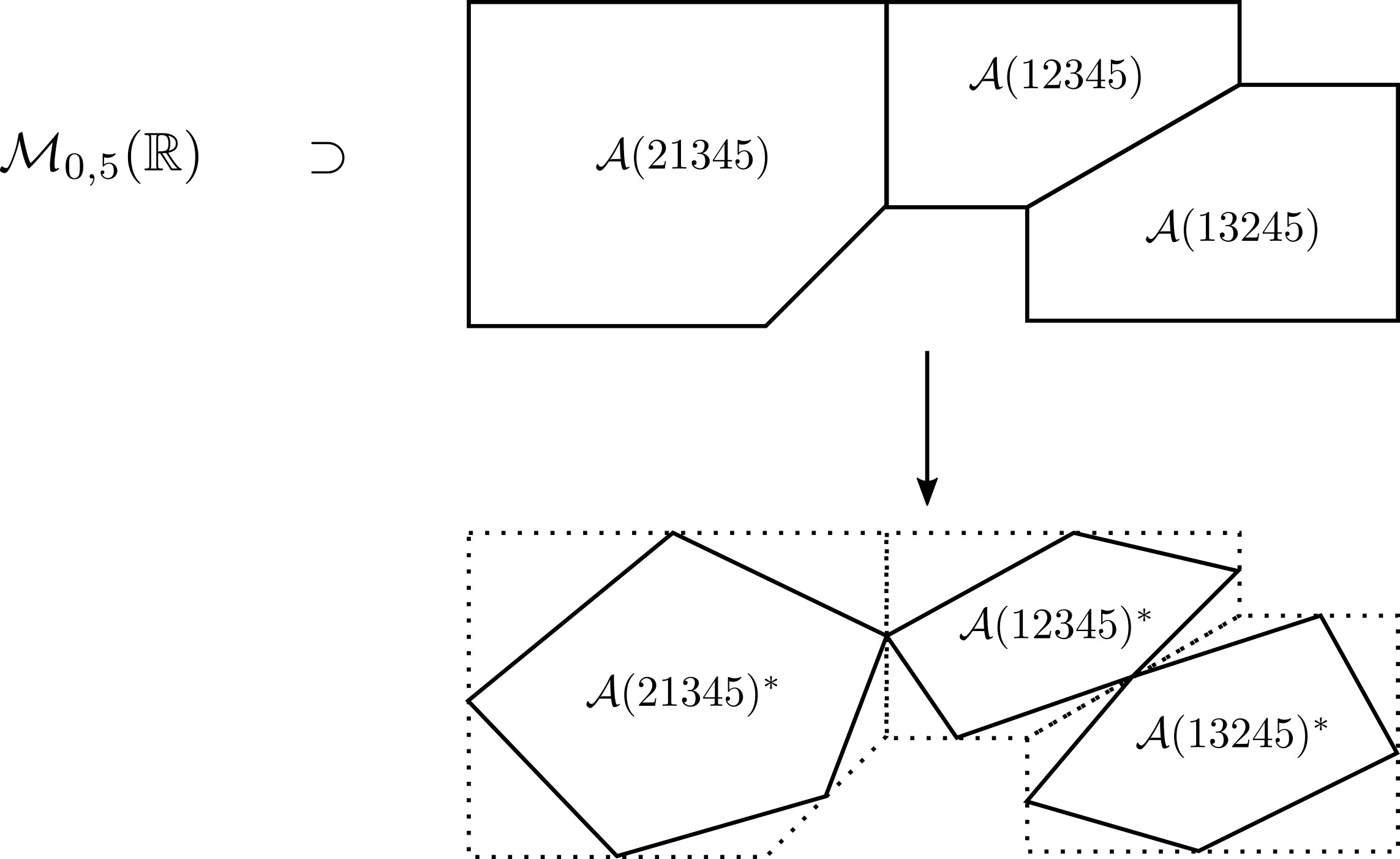}
\end{center}
\caption{A portion of $\overline{\mc{M}_{0,5}}(\mb{R})$. By associating to every side of each pentagon a vertex of the dual pentagon, we arrive at a net of dual pentagons with shared vertices. Rearranging this net leads us to figure \ref{ex5net}.}
\label{exb3}
\end{figure}
In general, $FN_n$ contains
$$
\sum_{k=2}^{n-2} {n-1\choose k} = 2^{n-1} - n-1
$$
vertices in $\mc{K}_n^*$. It is not clear whether $FN_n$ (or perhaps some double cover of it) can be mapped to the face lattice of a convex polytope. See section \ref{comments} for a discussion of this and related problems. 

\section{The partial amplitudes $m(\alpha|\beta)$}
\label{statement}
In the previous section we constructed an embedding of the (faces of the) dual associahedra $\mc{A}(\alpha)^*$ into dual kinematic space, $\mc{K}_n^*$. Given this embedding, our main claim is that the biadjoint scalar tree amplitudes are given by
$$
m(\alpha|\beta) = \Val (\p\mc{A}(\alpha)^* \cap \p\mc{A}(\beta)^*),
$$
where $\Val$ is a certain `valuation' or `volume' that we define below in equations \eqref{valdef1} and \eqref{valdef2}. We prove this formula as theorem \ref{theorem} at the end of this section. Before we define $\Val$, we need to review a few preliminaries.\footnote{See also appendix \ref{old} for a more detailed review and references.}

\subsection{Preliminaries}
\label{pre}
Let $P$ be a polyhedron in $\mc{K}_n^*$. For instance, $P$ could be a cone like $\mc{C}(\alpha)^*$. Then consider the integral
\begin{equation}
\label{Idef}
I_P(Z,\alpha') = \int\limits_P e^{-\alpha' W\cdot Z} \d W_P,
\end{equation}
where $\d W_P$ is the Euclidean volume form of the appropriate dimension to be integrated over $P$. We have introduced a parameter, $\alpha'$, which, at this stage, is not directly related to string theory in any way. Notice that $W$ is dimensionful and so $\alpha'$ is dimensionful, too. The integral $I_P$ is related to the volume of $P$. If $P$ is a bounded polytope, then the $(\alpha')^0$ term in the Laurent expansion of $I_P(Z,\alpha')$,
$$
I_P(Z,\alpha') = ...+  I_P^{(-1)}(Z)(\alpha')^{-1}+ I_P^{(0)}(Z) + I_P^{(1)}(Z)\alpha' + ...,
$$
will coincide with the volume of $P$. For now, let us fix $\alpha' = 1$. Suppose $P$ is a cone. Then, referring to equation \eqref{Idef}, we see that the integral $I_P(V,\alpha')$ is well defined if $V\cdot W > 0$ for all $W \in P$. This means that $I_P(Z,1)$ is well defined if $Z$ is in the dual cone $P^*\subset \mc{K}_n$.\footnote{The dual cone was defined in equation \eqref{dualconedef} in section \ref{reviewdual}. See also equation \eqref{olddualconedef} in appendix \ref{oldcones}.} If $Z$ lies on the boundary of $P^*$, then $W\cdot Z = 0$ for $W$ in some face of $P$. So $I_P(Z,1)$ diverges to $+\infty$ as $Z$ tends to the boundary of $P^*$ and the integral is not defined for $Z$ outside $P^*$. Now consider any polyhedral cone
$$
C = \Cone (W_1,...,W_k),
$$
where $k \leq \dim \mc{K}^*_n$. A standard calculation shows that for $Z$ in the interior of $C^*$, the integral $I_C$ is\footnote{See also theorem \ref{intcal2} in appendix \ref{oldint} where $I_C$ is computed for $k > \dim V$. The basic idea is to use the result that
$$
\int\limits_0^{\infty} \d x\, e^{-ax} = \frac{1}{a}
$$
several times.}
\begin{equation}
\label{mimic}
I_{C}(Z,1) = \frac{\avg{W_1,...,W_k}}{\prod_{i=1}^k (W_i\cdot Z)},
\end{equation}
where $\avg{W_1,...,W_k}$ is the Euclidean volume of the unit (open-closed) box
$$
\Box(W_1,...,W_k) = \left\{ \sum_{i=1}^kc_iW_i\,|\,0\leq c_i <1 \right\}.
$$
We are not interested in evaluating the integral $I$ for the cones $\mc{C}(\alpha)^*\subset \mc{K}_n^*$ introduced in section \ref{reviewdual}. For the purposes of writing amplitudes, the factors of $W_i\cdot Z$ appearing in equation \eqref{mimic} are propagators. The cone $\mc{C}(\alpha)^*$ has $n(n-3)/2$ vertices, and so $I_{\mc{C}(\alpha)^*}$ would contain too many propagators to be an $n$-point amplitude. An $n$-point amplitude should have $n-3$ propagators, since this is the number of internal lines in any $n$-point trivalent graph. All the combinatorial structure we need to write amplitudes is in the boundary $\p\mc{C}(\alpha)^*$. This boundary is generated by the facets of the dual associahedron $\mc{A}(\alpha)^*$ canonically embedded in $\mc{K}_n^*$. Under the canonical embedding, $\mc{A}(\alpha)^*$ is not necessarily itself a convex polytope in $\mc{K}_n^*$ (except for the $n=4$ case). However, the faces of $\mc{A}(\alpha)^* \subset \mc{K}_n^*$ are all convex polytopes by construction. It turns out that the amplitude $m(\alpha|\alpha)$ is given by a sum of terms in $I_{\p\mc{C}(\alpha)^*}$.

\subsection{The formula}
\label{new}
The integral $I_P$ computes biadjoint scalar amplitudes when $P$ is the cone generated by $\p\mc{A}(\alpha)^*$. To see this, consider the boundary $\p\mc{A}(\alpha)^*$ of the dual associahedron, canonically embedded in $\mc{K}_n^*$. This is the union of the embedded facets of $\mc{A}(\alpha)^*$, which we call $\{F_i\}$. So
\begin{equation}
\label{whamlicht}
\p\mc{A}(\alpha)^* = \bigcup_i F_i.
\end{equation}
Since $n-3$ is smaller than $\dim \mc{K}_n^* = n(n-3)/2$, we can use equation \eqref{mimic} to evaluate $I$ on $\p\mc{A}(\alpha)^*$. If $A$ and $B$ are cones that intersect in a strictly lower dimensional cone, then the standard properties of integrals tell us that
$$
I_{A\cup B} = I_A + I_B.
$$
It follows from this, and equation \eqref{whamlicht}, that
\begin{equation}
\label{valdef2}
I _{\Cone (\p\mc{A}(\alpha)^*)}(Z,1) = \sum_{\text{facets}~F} \frac{1}{\prod_{i=1}^{n-3} (W_{F(i)}\cdot Z)},
\end{equation}
where the facet $F$ is given by the convex hull of $W_{F(1)},...,W_{F(n-3)}$ and the volume of the unit box spanned by these vectors is one.\footnote{In abstract terms, we could normalise the volume so that the W vectors span unit-volume cells. Concretely, however, taking coordinates $Z = (s_{ij})$ and defining $W_S$ as in equation \eqref{gendef}, the vectors $W_{F(1)},...,W_{F(n-3)}$ span a parallelogram formed by translations of a unit box.} The right-hand-side of equation \ref{valdef2} is a biadjoint scalar tree amplitude. Indeed,
\begin{equation}
\label{lala}
m(\alpha|\alpha) = I _{\Cone (\p\mc{A}(\alpha)^*)}(Z,1).
\end{equation}
This is a special case of the theorem, theorem \ref{theorem}, that we prove below. We emphasise here that equation \eqref{lala} is merely an alternative form of the result presented for $m(\alpha|\alpha)$ by AHBHY. Their result is presented in sections 5.1 and 5.2 of \cite{ahbhy17}, and the equivalence with equation \eqref{lala} follows almost immediately from equation (5.7) in section 5.2 of their paper.\\

\begin{thm}
\label{theorem}
For dual associahedra $\mc{A}(\alpha)^*$ and $\mc{A}(\beta)^*$ embedded in $\mc{K}_n^*$ as described in section \ref{embed}, the biadjoint scalar amplitudes are given by
$$
m(\alpha|\beta) = I_{\Cone (\p\mc{A}(\alpha)^* \cap \p\mc{A}(\beta)^*)}(Z,1).
$$
\end{thm}
\begin{proof}
Given the set up as we have described it, this result is almost tautological. Vertices of $\mc{A}(\alpha)$ correspond to $\alpha$-planar trivalent graphs. So facets of $\mc{A}(\alpha)^*$ correspond to $\alpha$-planar trivalent graphs. Suppose that a particular $\alpha$-planar graph $g$ has propogators $1/S_{g(1)},....,1/S_{g(n-3)}$ where $-S_{g(i)} = W_{F(i)}\cdot V$ for some vectors $W_{F(i)} \in \mc{K}_n^*$. Then the corresponding facet $F$ of $\mc{A}(\alpha)^*$ is canonically embedded in $\mc{K}_n^*$ as
$$
\Conv(W_{F(1)},...,W_{F(n-3)}).
$$
If $g$ is also a $\beta$-planar graph, then $\Conv(W_{F(1)},...,W_{F(n-3)})$ will also be a facet of $\mc{A}(\beta)^*$, canonically embedded in $\mc{K}_n^*$. In this way, the $\alpha,\beta$-planar graphs are in 1-1 correspondence with shared facets of the dual associahedra $\mc{A}(\alpha)^*,\mc{A}(\beta)^*\subset \mc{K}_n^*$ under the canonical embedding. The result then follows from
$$
\sum_{\text{shared}~F} \frac{1}{\prod_{i=1}^{n-3} (W_{F(i)}\cdot V)} =\sum_{\alpha,\beta-\text{planar}~ g} \frac{1}{\prod_{i=1}^{n-3} (-S_{g(i)})}.
$$
This last expression is $m(\alpha|\beta)$.
\end{proof}

The statement of theorem \ref{theorem} emphasises the cones generated by the boundaries $\p\mc{A}(\alpha)^*$. However, the integral $I_{\Cone P}$ can also be regarded, if you prefer, as a volume or valuation of the polytope $P$ itself. To emphasise the role of the dual associahedra, we might define
\begin{equation}
\label{valdef1}
\Val(P)(Z) = I_{\Cone(P)}(Z,1).
\end{equation}
Notice that $\Val$ is a valuation (i.e. `like' a volume). This follows from the valuation property for $I_P$ ($I_{P\cup Q} = I_P + I_Q - I_{P\cap Q}$) and some other observations such as $\Cone(P\cup Q) =  \Cone(P) \cup \Cone(Q)$. This said, we can write the following.\\

\begin{cor}
\label{valley}
The biadjoint scalar tree amplitudes are given by
$$
m(\alpha|\beta) = \Val(\p\mc{A}(\alpha)^* \cap \p\mc{A}(\beta)^*),
$$
where $\Val$ is the valuation defined in equation \eqref{valdef1}.
\end{cor}

\subsection{Examples}
\label{examples}
We now illustrate theorem \ref{theorem} at four and five points. No essentially new phenomena appear at higher points, so these examples suffice to illustrate the result.\\

\begin{ex}
The kinematic associahedra at four points were presented explicitly in section \ref{embed}, equations \eqref{ex4da} to \eqref{ex4da2}. For example, the dual associahedron $\mc{A}(1234)^*$ is given by
$$
\mc{A}(1234)^* = \Conv (W_{12},W_{23}),
$$
and it has two zero-dimensional facets, namely $W_{12}$ and $W_{23}$. Then we compute
$$
I_{\Cone(W_{12})\cup\Cone(W_{23})} = \frac{1}{ W_{12}\cdot V} + \frac{1}{ W_{23}\cdot V}.
$$
It follows that the amplitude is
$$
m(1234|1234) =\frac{1}{-s_{12}} + \frac{1}{-s_{23}}.
$$\\
\end{ex}

\begin{ex}
Recall from equation \eqref{ex4da2} that
$$
\mc{A}(3124)^* = \Conv (W_{13},W_{12}).
$$
Then we evaluate the intersection
$$
\mc{A}(1234)^* \cap \mc{A}(3124)^* = W_{12}.
$$
This implies the amplitude
$$
m(1234|3124) = \Val (\Cone(W_{12})) = \frac{1}{W_{12}\cdot V} = \frac{1}{-s_{12}}.
$$\\
\end{ex}
\begin{figure}
\begin{center}
\includegraphics[scale=0.40]{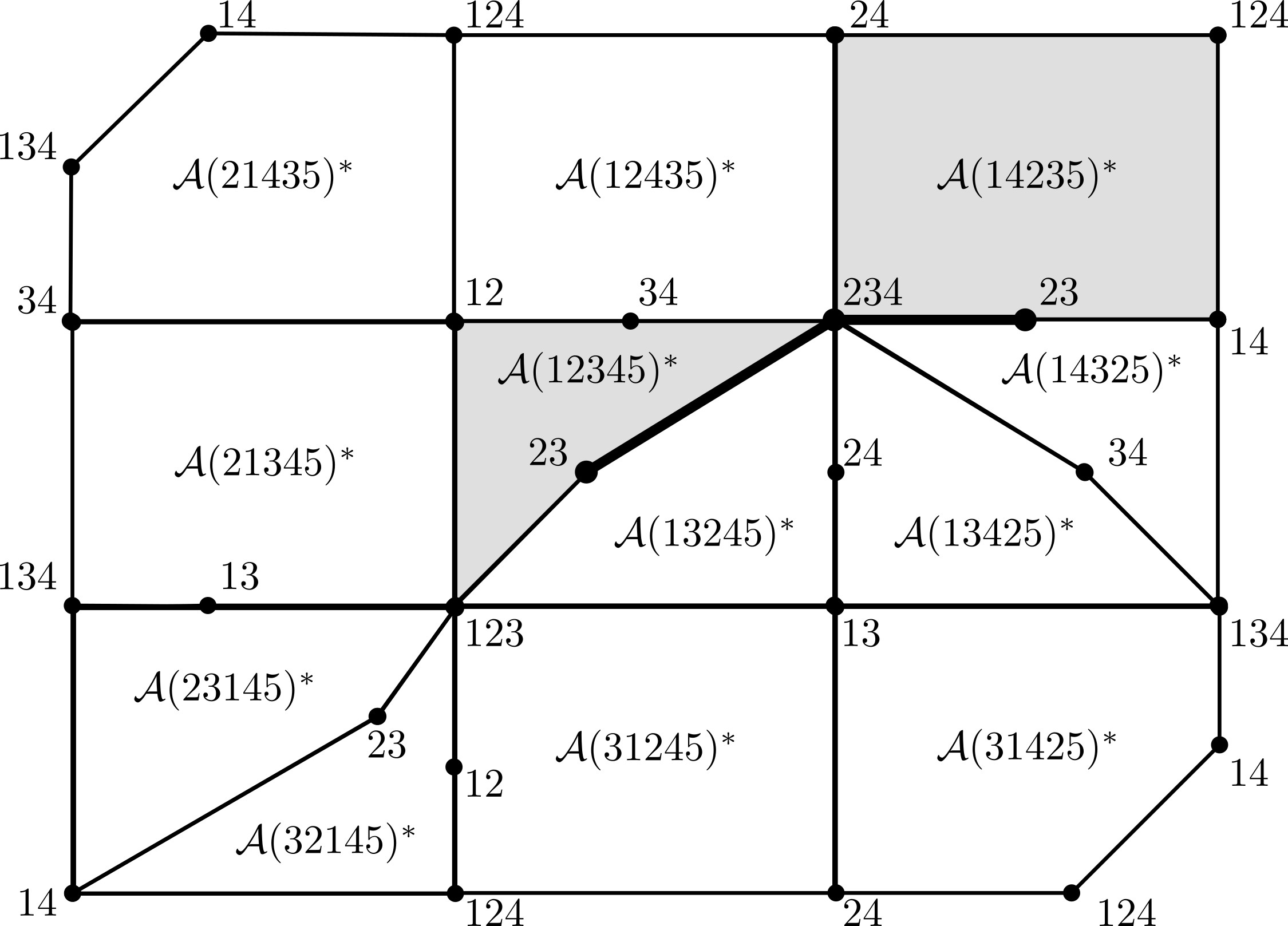}
\end{center}
\caption{The shared facet of $\mc{A}(12345)^*$ and $\mc{A}(14235)^*$. Vertices in the net with the same label are identified under the canonical embedding into $\mc{K}_5^*$.}
\label{ex5a}
\end{figure}
\begin{ex}
At five points, the kinematic space $\mc{K}_5$ is defined by the hyperplane
$$
s_{12}+s_{13}+s_{14}+s_{23}+s_{24}+s_{34}=0
$$
in $V = \mb{R}^6$. Consider the ordering $\alpha = 12345$. The cone $\mc{C}(\alpha)$ is defined by the inequalities
\begin{align*}
&X_{13} = -s_{12} \geq 0\\
&X_{14} = -s_{123} \geq 0\\
&X_{24} = -s_{23} \geq 0\\
&X_{25} = -s_{234} \geq 0\\
&X_{35} = -s_{34} \geq 0.\\
\end{align*}
The dual cone $\mc{C}(\alpha)^*$ is the conic hull
$$
\mc{C}(\alpha)^* = \Cone (W_{12},W_{123},W_{23},W_{234},W_{34}).
$$
(For the definition of the dual vectors $W_S$, see equation \eqref{gendef} in section \ref{reviewdual}.) Similarly, the dual cone $\mc{C}(\beta)^*$ for $\beta = 14235$ is the conic hull
$$
\mc{C}(\beta)^* = \Cone (W_{14},W_{124},W_{24},W_{234},W_{23}).
$$
The associated dual associahedra, $\mc{A}(\alpha)^*$ and $\mc{A}(\beta)^*$, canonically embedded in $\mc{K}_5^*$, share one face:
$$
\mc{A}(\alpha)^*\cap \mc{A}(\beta)^* = \Conv (W_{23},W_{234}).
$$
This can be read off from figure \ref{ex5net}, or computed explicitly by listing the faces of the two dual associahedra. See figure \ref{ex5a} for a drawing that highlights the intersection. We compute
$$
I_{\Cone(W_{23}W_{234})}(Z,1) = \frac{1}{Z\cdot W_{23} Z \cdot W_{234}}.
$$
So the amplitude is
$$
m(\alpha|\beta) = \Val \left(\Conv (W_{23},W_{234})\right) = \frac{1}{s_{23}s_{234}}.
$$\\
\end{ex}
\begin{figure}
\begin{center}
\includegraphics[scale=0.40]{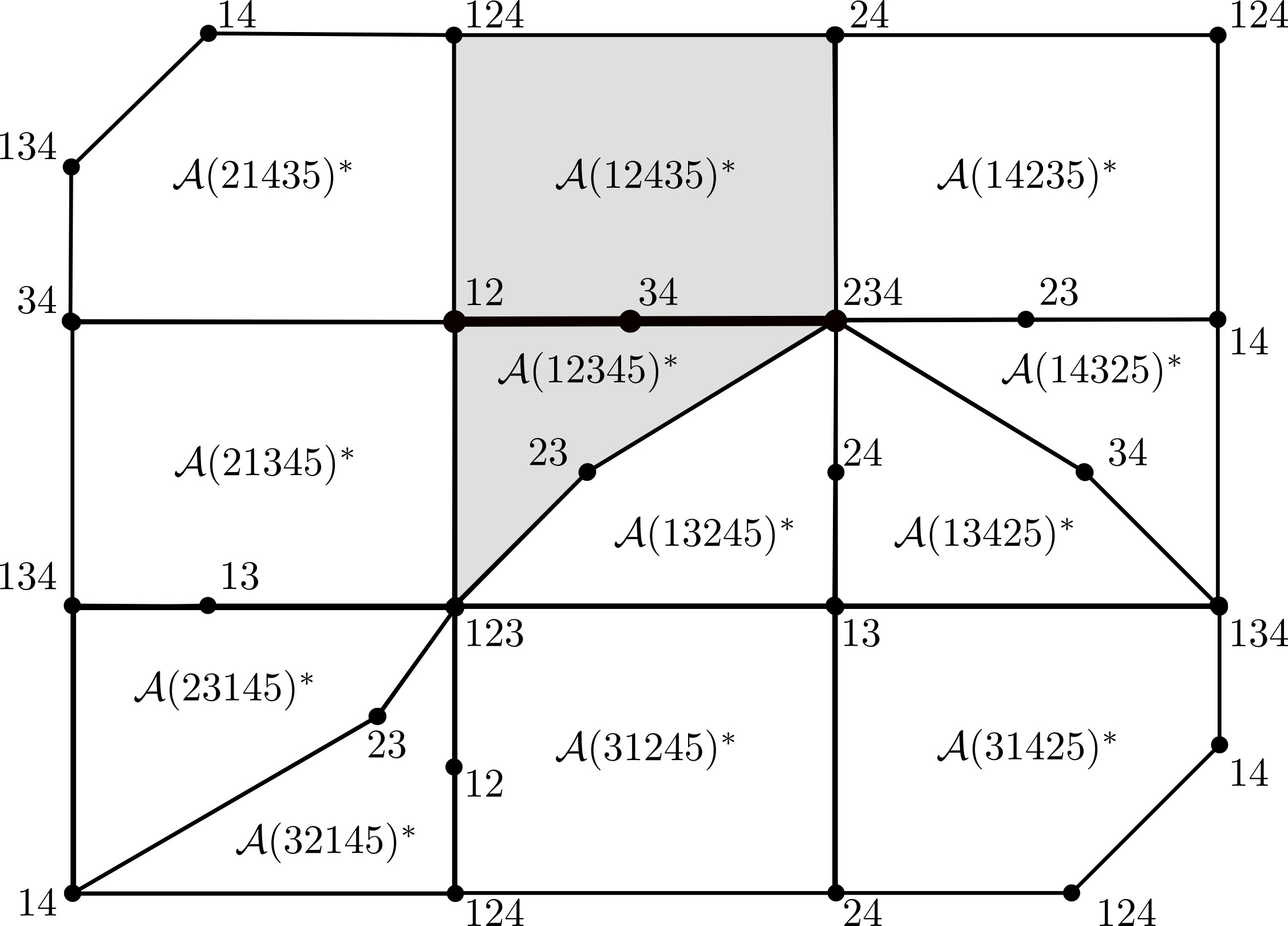}
\end{center}
\caption{The shared facets of $\mc{A}(12345)^*$ and $\mc{A}(12435)^*$, pictured as a net. Vertices with the same label are identified under the canonical embedding into $\mc{K}_5^*$.}
\label{ex5b}
\end{figure}
\begin{ex}
\label{sibilanga}
To give another example, the dual cone for $\gamma = 12435$ is the conic hull
$$
\mc{C}(\gamma)^* = \Cone (W_{12},W_{124},W_{24},W_{234},W_{34}).
$$
The dual associahedron $\mc{A}(\gamma)^*$, canonically embedded in $\mc{K}_5^*$, shares two faces with $\mc{A}(\alpha)^*$. This can be read off from figure \ref{ex5net} or computed by hand. The shared faces are
$$
\text{Int} = \mc{A}(\alpha)^*\cap\mc{A}(\gamma)^* = \Conv(W_{234},W_{34})\cup \Conv (W_{34},W_{12}).
$$
The two faces are highlighted in figure \ref{ex5b}. The valuation $I$ of this intersection is
$$
I_{\Cone(\text{Int})}(Z,1) = \frac{1}{W_{34}\cdot Z W_{234}\cdot Z} + \frac{1}{W_{34}\cdot Z W_{12}\cdot Z}.
$$
The amplitude is
$$
m(\alpha|\gamma) = \Val (\text{Int}) = \frac{1}{s_{23}s_{234}} + \frac{1}{s_{12}s_{34}}.
$$
\\
\end{ex}
\begin{ex}
As a final example, the cone $\mc{C}(\delta)^*$ for $\delta = 13524$ is given by
$$
\mc{C}(\delta)^* = \Cone (W_{13},W_{24},W_{124},W_{14},W_{134}).
$$
One can readily verify that $\mc{C}(\delta)^*\cap \mc{C}(\alpha)^* = \emptyset$.\footnote{This is an exercise in linear reduction. One should write the generators of $\mc{C}(\delta)^*$ as linear combinations of the generators of $\mc{C}(\alpha)^*$ and observe that each generator of $\mc{C}(\delta)^*$ lies outside of $\mc{C}(\alpha)^*$. } This means that the dual associahedra $\mc{A}(\alpha)^*$ and $\mc{A}(\delta)^*$ share no common faces in the canonical embedding into $\mc{K}_5^*$. (This can also be read off from figure \ref{ex5net}.) This means that the amplitude vanishes:
$$
m(\alpha|\delta) = \Val (\mc{A}(\alpha)^*\cap \mc{A}(\delta)^*) = 0.
$$
\end{ex}

\section{A connection with the KLT kernel}
\label{statement'}
In this section we state a connection between the dual associahedra in $\mc{K}_n^*$ and the inverse KLT kernel. We denote the inverse KLT kernel by $m_{\alpha'}(\alpha|\beta)$, following Mizera. \cite{mizera1706} In section \ref{stringyformula} we present a new formula for the diagonal elements $m_{\alpha'}(\alpha|\alpha)$. This formula involves a discrete sum over a lattice which we introduce in section \ref{lattice}. In kinematic space, $\mc{K}_n$, the lattice is a standard lattice of points for which the Mandelstam variables $s_{ij}$ take integer values (or, if you prefer, are integer multiples of $1/\alpha'$). Before we prove the formula in section \ref{stringyformula}, we review some technical preliminaries in sections \ref{lattice} and \ref{stringyreview}. But to put this discussion in context, let us begin by recalling that Kawai-Lewellen-Tye derived a relation of the form
$$
A_{\text{closed}} = \sum_{\alpha,\beta} A_{\text{open}}(\alpha) \text{KLT}(\alpha|\beta) A_{\text{open}}(\beta)
$$
between closed string tree amplitudes and open string amplitudes. This relation was derived in \cite{klt} using an analytic continuation argument to deform the integration contours of the amplitudes in such a way that $A_{\text{closed}}$ factorises as shown. The kernel, $\text{KLT}(\alpha|\beta)$, can be inferred given knowledge of the string amplitudes. However, in \cite{mizera1706}, Mizera conjectured an algorithmic description of the inverse kernel, $m_{\alpha'}(\alpha|\beta)$, that makes no reference to the string amplitudes (and their associated hypergeometric functions). In subsequent work, Mizera showed that his formulas for $m_{\alpha'}(\alpha|\beta)$ can be regarded as (twisted) intersection pairings of the associahedra that tile $\mc{M}_{0,n}(\mb{R})$. \cite{mizera1708} For this reason, it is interesting that $m_{\alpha'}(\alpha|\alpha)$ appears naturally in our present context, where we are concerned with intersecting dual associahedra. For further speculations about whether the two presentations are related, see section \ref{comments}.

\subsection{The lattice}
\label{lattice}
\begin{figure}
\begin{center}
\includegraphics[scale=0.40]{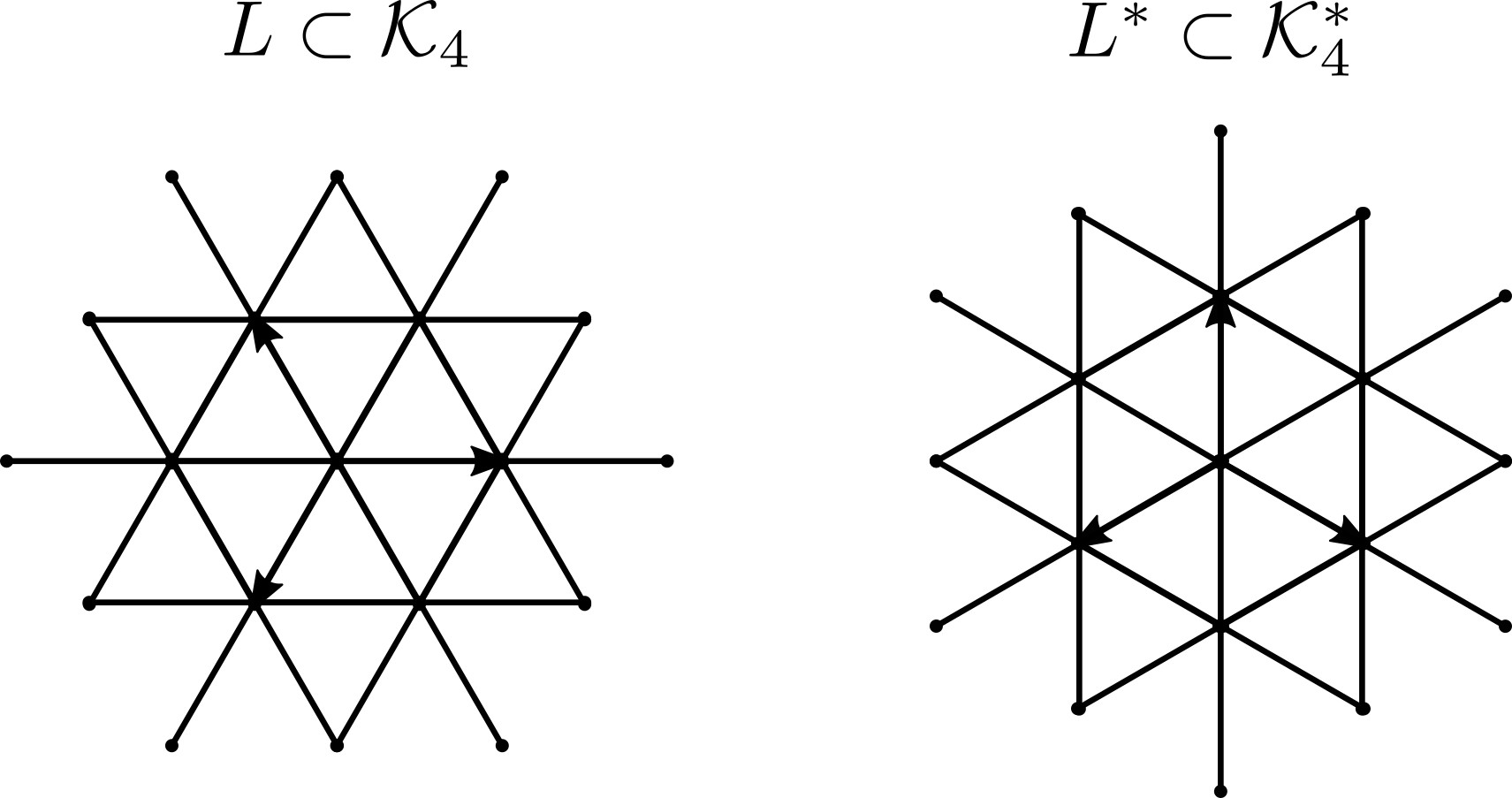}
\end{center}
\caption{The lattice of points with integer Mandelstam variables, $L$, and its dual lattice, $L^*$.}
\label{ex4lat}
\end{figure}
Recall from section \ref{review} that $\mc{V}_n = \mb{R}^{n(n-3)/2+1}$ is the vector space with coordinates $s_{ij}$ for all pairs $i,j$ strictly less than $n$. Let $\tilde L$ be the standard lattice $\mb{Z}^{n(n-3)/2+1} \subset \mc{V}_n$ of points for which the $s_{ij}$ are all integers. Given this, we have a lattice $L = \tilde L \cap \mc{K}_n$ in kinematic space.\footnote{It is clear that $\mc{K}_n$ is a rational subspace with respect to $L$. That is, we can choose a basis that spans $\mc{K}_n$ from among vectors lying in $L$. Indeed, for some ordering of the $s_{ij}$ we could adopt a basis of the form
$$
\{(1,-1,0,0,...),(0,1,-1,0,...),(0,0,1,-1,...),...\},
$$
which are all vectors in $L$ and span the hyperplane $\mc{K}_n$.} It turns out that one can choose generators for the cones $\mc{C}(\alpha)$ from among vectors in the lattice $L$.\footnote{The proof is as follows. The vectors $W_S$ defined in equation \eqref{gendef} are $L^*$-vectors, as shown in equation \eqref{show} of the main text. This means that the hyperplanes $W_S\cdot Z = 0$ in $\mc{K}_n$ are $L$-rational. The rays of the cone $\mc{C}(\alpha)$ are given by intersecting these hyperplanes. It follows that the rays are $L$-rational and, therefore, they are generated by some vector in $L$.} For this reason, the cones $\mc{C}(\alpha)$ are called `rational' cones, with respect to the lattice $L$. Dual to $L$, we have the lattice $L^*$ in $\mc{K}_n^*$. The cones $\mc{C}(\alpha)^*$ are generated by vectors in $L^*$. Indeed, the vectors $W_S$ that we have been using as generators for $\mc{C}(\alpha)^*$ are $L^*$-vectors. This is because
\begin{equation}
\label{show}
Z\cdot W_S \in \mb{Z} \qquad\text{for all}\qquad Z\in L,
\end{equation}
which is the condition for $W_S$ to be a vector in $L^*$. To illustrate all of this, consider again the example of four points. A basis for the lattice $L$ in $\mc{K}_3 \subset \mb{R}^3$ is
$$
L = \mb{Z} \left( \begin{bmatrix} 1 \\ 0 \\ -1 \end{bmatrix}, \begin{bmatrix} 0 \\ 1 \\ -1 \end{bmatrix} \right).
$$
Alternatively, we can employ the coordinates on $\mc{K}_3$ that we introduced in equation \eqref{hexa} (section \ref{reviewcones}). In these coordinates,
$$
L = \mb{Z} \left( \begin{bmatrix} 1/2 \\ \sqrt{3}/2  \end{bmatrix}, \begin{bmatrix} 1 \\ 0 \end{bmatrix} \right).
$$
This is just the triangular lattice, see figure \ref{ex4lat}.  On the other hand, the dual lattice is
$$
L^* = \mb{Z} \left( W_{12}, W_{23} \right),
$$
which is the triangular lattice rotated by $\pi/6$. Notice that $W_{13} = - W_{12} - W_{13}$, and so $W_{13}$ is also a vector in $L^*$.

\subsection{Preliminaries}
\label{stringyreview}
Before we can present our formula for $m_{\alpha'}(\alpha|\alpha)$ in theorem \ref{theorem'}, we need to introduce a sum over lattice points which is the direct analogy of the integral $I_C$ that we introduced in equation \eqref{Idef} (section \ref{pre}). Let $P$ be a polyhedron in $\mc{K}_n^*$. We are interested in the sum
$$
S_P(V,\alpha') = \sum_{W\in P\cap L^*} e^{-\alpha' W\cdot V}.
$$
Just as the integral $I_P$ is the related to the volume of $P$, the sum $S_P$ is related to the number of $L^*$-points contained in $P$. If $P$ is a bounded polytope, the number of $L^*$-points would be given by
$$
\#(P\cap L^*) = \lim_{\alpha'\rightarrow 0} S_P(V,\alpha').
$$
For now, let $\alpha'$ be positive and real. When $P$ is a cone, the sum $S_P$ is well defined if $V$ is in the interior of the dual cone $P^*$.\footnote{Recall that, for such a $V$, $W\cdot V\geq 0$ for all $W \in P$. The convergence of the sum $S_P(V,\alpha')$ then follows from the usual convergence statement:
$$
\sum_{n=0}^{\infty} x^n
$$
converges to $1/(1-x)$ if $|x| < 1$.} Now fix some cone
$$
C = \Cone(W_1,...,W_k),
$$
where $k \leq \dim \mc{K}_n^*$. The sum is particularly easy to evaluate if the vectors $W_i$ are all generators of the lattice $L^*$. One finds
\begin{equation}
\label{key}
S_C(V,\alpha') = \prod_{i=1}^k \frac{1}{1-e^{-\alpha' W_i\cdot V}},
\end{equation}
for $V \in C^*$.\footnote{See appendix \ref{appdiscrete}, for the analogous formulas when $W_i$ are not lattice generators and for when $k>\dim \mc{K}_n^*$. The general idea is to repeatedly apply
$$
\sum_{n=0}^{\infty} e^{-n\alpha' W_i\cdot Z} = \frac{1}{1-e^{-\alpha'W_i\cdot Z}}.
$$} In section \ref{statement} we studied the integral $I_{\p C}$ over the boundary, $\p C$, of a cone $C$. We conclude this subsection by evaluating the discrete sum $S_{\p C}$. We find a large number of new terms, when compared with $I_{\p C}$, because of the valuation property of the sum $S_P$. In effect, the sum $S_P$ `sees' not just the faces $\Cone(F)$ of the cone $C$, but also their intersections $\Cone(F_i)\cap \Cone(F_j)$. This turns out to be crucial for the realtionship with the inverse KLT kernel, which we explain in section \ref{stringyformula}.\\

\begin{lem}
\label{lemma'}
Let $C$ be a cone in the vector space $\mc{K}^*$ whose boundary is generated by some set of polytopes $\{F\}$ of the same dimension that intersect in polytopes of strictly lower dimension. Then
$$
S_{\p C}(Z,\alpha') = - \sum_{k=1} (-1)^k \sum_{F_1,...,F_k} S_{\Cone (F_1\cap ... \cap F_k)}(Z,\alpha'),
$$
where the summations are over all $k$-tuples, $\{F_1,...,F_k\}$, such that the intersection $F_1\cap ... \cap F_k$ is not empty.\\
\end{lem}
\begin{proof}
The sum $S_P$ is a `valuation' in the sense that it behaves like a volume:
\begin{equation}
\label{valprop}
S_{P\cup Q} = S_P + S_Q - S_{P\cap Q},
\end{equation}
for two polyhedra $P$ and $Q$. The lemma is a consequence of this property, writing
$$
\p C = \bigcup_i \Cone(F_i).
$$
To arrive at the formula explicitly is an exercise in induction, making use of the observation that
$$
\Cone(F_i)\cap\Cone(F_j) = \Cone(F_i\cap F_j)
$$
and using standard set theory identities such as $(A\cup B)\cap X = (A\cap X)\cup(B\cap X)$.
\end{proof}

\subsection{The formula}
\label{stringyformula}
We will now observe that the sum $S_{\p \mc{C}(\alpha)}$, defined above, is related to the diagonal components $m_{\alpha'}(\alpha|\alpha)$ of the inverse KLT kernel.\\

\begin{thm}
\label{theorem'}
For a dual associahedron $\mc{A}(\alpha)^*$ canonically embedded in $\mc{K}_n^*$ (as described in section \ref{review}),
$$
S_{\Cone (\p \mc{A}(\alpha)^*)}(Z,2\pi i \alpha') = 1- m_{\alpha'}(\alpha|\alpha).
$$
The left hand side can be found by, for instance, evaluating $S_{\Cone \p \mc{A}(\alpha)^*}(Z,\alpha')$ and then analytically continuing the result by replacing $\alpha'$ with $2\pi i \alpha'$.
\end{thm}
\begin{proof}
This follows directly from lemma \ref{lemma'}. Indeed, combining the lemma with equation \eqref{key},
$$
1-S_{\Cone (\p \mc{A}(\alpha)^*)}(Z,\alpha') = 1 + \sum_{k=1} \sum_{F_1,...,F_k} \prod_{i=1}^{n-2-k} \frac{1}{e^{-\alpha' W_{F(i)}\cdot Z}-1}.
$$
Since the sum only includes $k$-tuples $(F_1,...,F_k)$ such that $F_1\cap...\cap F_k \neq \emptyset$, any $F_1\cap...\cap F_k$ appearing in the sum has $n-2-k$ vertices $W_{I_a}$. Then $W_{I_a}\cdot V = - S_{I_a}$, where $S_{I_a}$ are the $n-2-k$ propagators associated to the codimension-$k$ face $F_1\cap ... \cap F_k$. (Dually, they are associated to a codimension-$(n-2-k)$ face in $\mc{A}(\alpha)$.) Then
$$
1-S_{\Cone (\p \mc{A}(\alpha)^*)}(Z,2\pi i\alpha') = 1 + \sum_{k=1} \sum_{F_1,...,F_k} \prod_{a=1}^{n-2-k} \frac{1}{e^{2\pi i\alpha' S_{I_a}\cdot Z}-1}.
$$
This is the expression for $m_{\alpha'}(\alpha|\alpha)$ discussed by Mizera in \cite{mizera1708} (equation (4.19) of Mizera's paper). The formula first appears in \cite{ky94}.
\end{proof}

We emphasize here that we have not found a natural interpretation for the off-diagonal elements $m_{\alpha'}(\alpha|\beta)$ in terms of the lattice sum. See section \ref{comments} for some speculations. To illustrate theorem \ref{theorem'}, consider the four point example. For instance,
$$
\p \mc{A}(1234)^* = \{W_{12}\}\cup\{W_{23}\}.
$$
It follows that
$$
1 - S_{\Cone(W_{12})} - S_{\Cone(W_{23})}  = 1 + \frac{1}{e^{-\alpha' s_{12}} - 1} + \frac{1}{e^{-\alpha' s_{23}} - 1}.
$$
After an analytic continuation, $\alpha' \mapsto 2 \pi i \alpha'$, this expression becomes
$$
m_{\alpha'}(1234|1234) = - \frac{1}{2i\tan (\pi \alpha' s_{12})} - \frac{1}{2i\tan (\pi \alpha' s_{23})},
$$
which is the formula also given by Mizera. We recover the amplitude $m(1234|1234)$ from the pole in $\alpha'$,
$$
m(1234) = \oint \d\alpha' m_{\alpha'}(1234) = \frac{1}{-s_{12}} + \frac{1}{-s_{23}},
$$
which follows since $1/\tan(x) \simeq 1/x + \mc{O}(x^0)$.

\section{Further comments}
\label{comments}
The main result in this paper is the formula for $m(\alpha|\beta)$ presented in section \ref{statement}, theorem \ref{theorem}. The formula is based on an embedding of dual associahedra $\mc{A}(\alpha)^*$ into dual kinematic space $\mc{K}_n^*$. The amplitude $m(\alpha|\beta)$ can then be expressed in terms of the shared faces of $\mc{A}(\alpha)^*$ and $\mc{A}(\beta)^*$. Arguably, what we have done is ``trivial'' in the sense that the dual associahedra do no more than express Feynman diagrammatics in a geometric setting. The intersecting faces of the dual associahedra are just a fancy way to describe a sum over trivalent graphs which are $(\alpha,\beta)$-planar. There is, however, something interesting about the new presentation. The embedded associahedra $\mc{A}(\alpha)^*$ tile a larger object in dual kinematic space. This is the object pictured in figure \ref{ex4c} for four points, where it is a triangle, and in figure \ref{ex5net} for five points, where it resembles a degenerate or `halved' dodecahedron. What is this object in general? And, given the role that $m(\alpha|\beta)$ plays in the double copy relation, how can this object be related to Yang-Mills and gravity amplitudes? We offer some speculations below, and describe many unresolved loose ends.

\paragraph{The open string moduli space.}
In section \ref{embed} we introduced $(n-1)!/2$ embeddings of the boundary of the dual associahedron into $\mc{K}_n^*$. We defined the union of these embeddings,
\begin{equation}
\label{ppp}
FN_n = \bigcup_{\alpha}\,\p\mc{A}(\alpha)^* \subset \mc{K}_n^*.
\end{equation}
$FN_n$ does not bound a convex polytope in $\mc{K}_n^*$, except when $n=4$ (in which case it bounds an equilateral triangle). The first loose end is to study $FN_n$ in its own right and explain how it is related to $\mc{M}_{0,n}(\mb{R})$.  The open string moduli space, $\mc{M}_{0,n}(\mb{R})$, after compactification, is tiled by $(n-1)!/2$ copies of the associahedron. (See, for instance, theorem 3.1.3 of \cite{devadoss}.) There is a duality map from each associahedron $\mc{A}(\alpha)$ in $\mc{M}_{0,n}(\mb{R})$ to the corresponding dual associahedron $\mc{A}(\alpha)^*$. In this sense, $FN_n$ is something like a `dual' of the open string moduli space. With some effort, this could probably be made into a precise statement about dual polytopes.\footnote{A possible route would be to exploit the relationship of  $\mc{M}_{0,n}(\mb{R})$ to the permutoassociahedron (introduced by Kapranov in \cite{kapranov93}) which can be realised as a convex polytope, as shown in \cite{zr94}. As a convex polytope, we can take the polytope dual of the permutoassociahedron, from which we might be able to recover $FN_n$.} A duality statement might help us to present a relationship between the intersecting dual associahedra in this paper and the (twisted) intersections of associahedra that have been explored recently by Mizera in \cite{mizera1708}.

\paragraph{The permutohedron.}
\begin{figure}
\begin{center}
\includegraphics[scale=0.40]{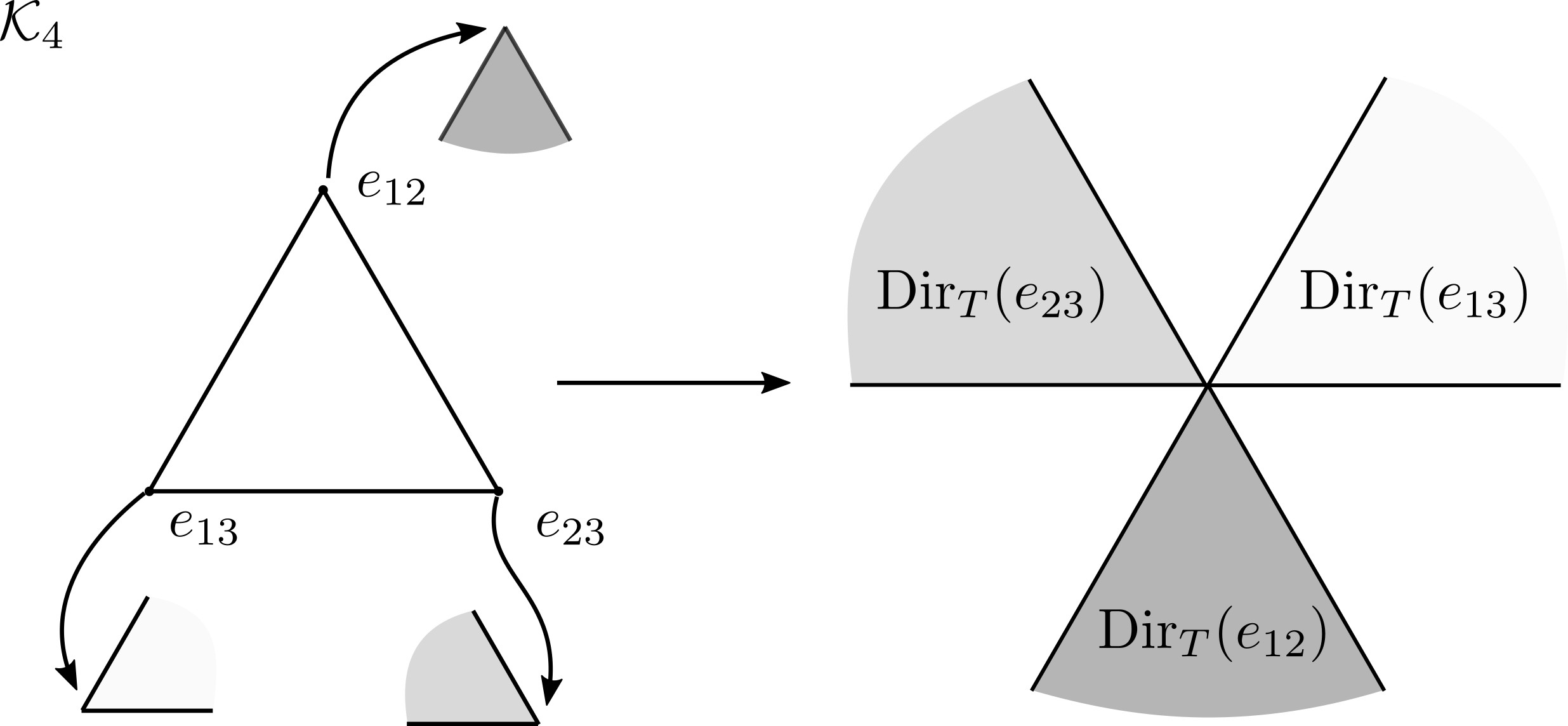}
\end{center}
\caption{At four points, the cones $\mc{C}(\alpha)$ defined by AHBHY arise as the tangent cones to a triangle. In conjecture \ref{conjecture}, we speculate that this has a generalisation to higher points.}
\label{exd1}
\end{figure}
A second loose end concerns the cones $\mc{C}(\alpha)$ in kinematic space $\mc{K}_n$ introduced by AHBHY. There are $(n-1)!/2$ cones cut out by
$$
2^{n-1} - n -1
$$
hyperplanes. At four points, the three cones, $\mc{C}(\alpha)$, are the `tangent cones' of a triangle.\footnote{In general, the tangent space at the vertex of a polytope is a cone. See section \ref{oldcones} and especially equation \eqref{dirrr} for the definitions.} Indeed, consider the triangle 
$$
T_4 = \Conv \left(e_{12},e_{13},e_{23}\right)
$$
in $\mc{K}_4$. The vectors $e_{12}$ and $e_{23}$ are as in equation \eqref{hexa}, while $e_{13}$ is defined as $-e_{12}-e_{23}$. The three vertices of $T_4$ have tangent cones such as, at the vertex $e_{12}$,
$$
\Cone(e_{13}-e_{12},e_{23}-e_{12}).
$$
But this cone is just the cone $\mc{C}(2314)$ as defined by AHBHY. Indeed, we see that $\mc{C}(1234)$ is the tangent cone of $T_4$ at $e_{23}$ and $\mc{C}(3124)$ is the tangent cone of $T_4$ at $e_{13}$. We illustrate the idea in figure \ref{exd1}. This leads us to a conjecture.\\

\begin{con}
\label{conjecture}
AHBHY's cones $\mc{C}(\alpha) \subset \mc{K}_n$ arise as the tangent cones of a polytope $P_n$ with $(n-1)!/2$ vertices.\\
\end{con}
\begin{figure}
\begin{center}
\includegraphics[scale=0.40]{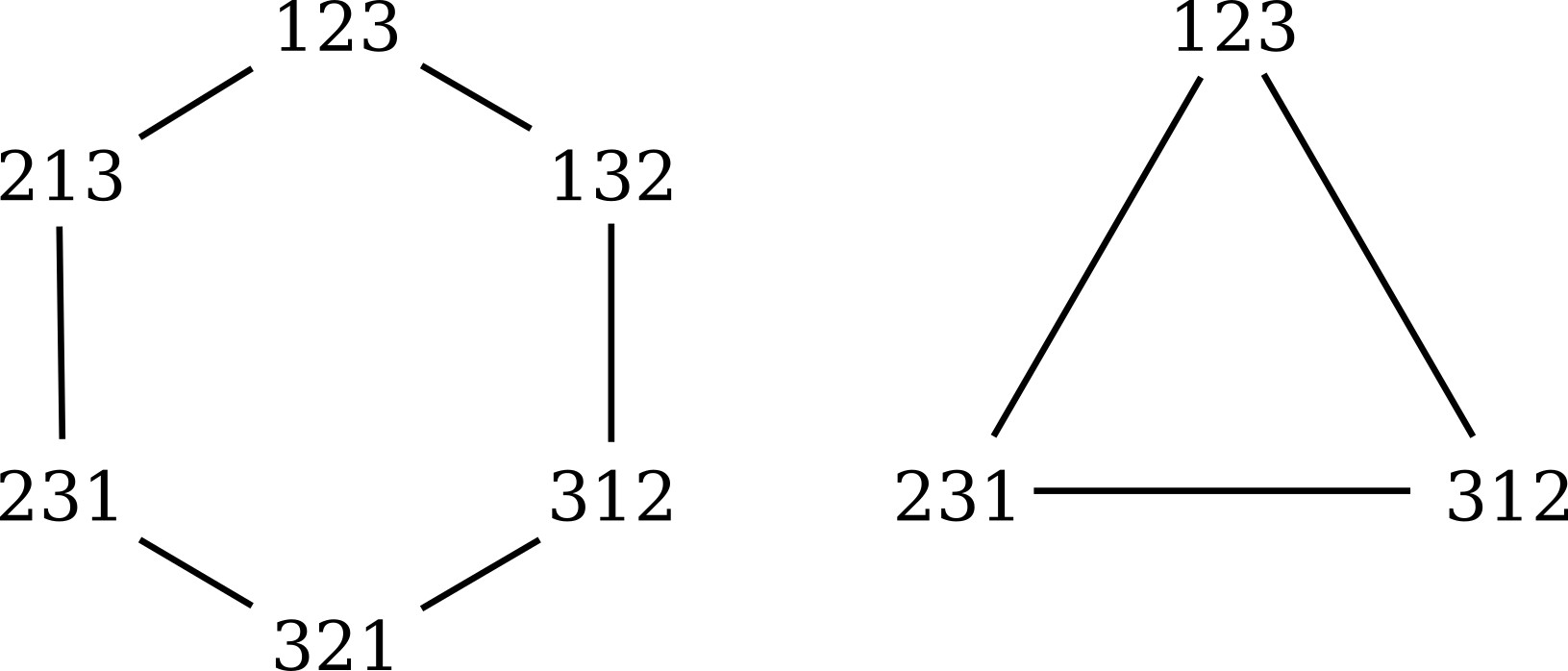}
\end{center}
\caption{The permutohedron on three letters is the hexagon shown on the left. This can be collapsed into a `degenerate' or `generalised' permutohedron on three letters which ignores flips. This is the triangle on the right.}
\label{exd2}
\end{figure}
A polytope in $N$ dimensions must have at least $N+1$ vertices. So, in support of the conjecture, notice that $(n-1)!/2$ is greater than  $\dim \mc{K}_n$ for all $n \geq 4$. If the polytopes $P_n$ do exist, it is not clear whether or not they should be convex polytopes. At four points, the triangle is a convex polytope. We do not know if this generalises. A possible way to test convexity is theorem \ref{dualcones} in appendix \ref{oldcones}. This theorem states that the union of dual tangent cones of a convex polytope covers all of the dual vector space. In our case, this would mean that
\begin{equation}
\label{provemewrong}
\bigcup_{\alpha} \mc{C}(\alpha)^* \simeq \mc{K}_n^*.
\end{equation}
This is certainly true for $n=4$ (see figure \ref{ex4c}). But if equation \eqref{provemewrong} is false for any $n$, this would mean that the polytopes $P_n$ are not always convex. A possible candidate for the polytopes $P_n$ are generalised permutohedra. The permutohedron can  be realised by considering the orbit of the symmetric group $\mf{G}_n$ on a generic point in $\mb{R}^n$. That is, for a point $(x_1,...,x_n) \in \mb{R}^n$, consider all the points
$$
(x_{\sigma(1)},...,x_{\sigma(n)})
$$
obtained for permutations $\sigma$. Taking the convex hull of these points gives a polytope, which is the permutohedron. For instance, the permutohedron on three letters is the hexagon. Indeed, consider the point $(c,-c,0)$ in $\mc{V}_3$. Under the symmetric group action, the orbit of this point are the six points in $\mc{K}_3$ at the boundaries of the associahedra $\mc{A}(\alpha)$. The convex hull of these points is a hexagon (refer back to figure \ref{ex1a}, for instance).\footnote{It is clearly a general fact that the orbit of a point $(\lambda_1,...,\lambda_n)$ lies in the hyperplane $\lambda_1+...+\lambda_n = \text{const.}$. So the orbits of points in the hyperplane $\mc{K}_n$ remain in $\mc{K}_n$.} `Generalised' permutohedra can be obtained by taking the orbits of non-generic points (for which some of the coordinates are equal to each other) or by translating the faces. \cite{postnikov05} For instance, we could take the orbit of the point $(1,1,-2)$ in $\mc{K}_4$, which gives a triangle. (We show these two permutohedra, the hexagon and the triangle, in figure \ref{exd2}.) At higher points, one can likewise obtain generalised permutohedra on $n-1$ letters that have the required number of $(n-1)!/2$ vertices. The polytopes $P_n$ conjectured in conjecture \ref{conjecture} may well be realisations of these generalised permutohedra.\footnote{A result of reference \cite{lmh} says that $P_n$ is a generalised permutohedron if its face lattice refines the Weyl fan. For work on expanding cones as a sum of Weyl fans (or their duals, called `plates'), see N. Early's recent paper. \cite{early}}

\paragraph{The double copy.}
The amplitudes $m(\alpha|\beta)$ form the inverse kernel for the field theory double copy relation
\begin{equation}
\label{ftklt}
M_n = \sum_{\alpha,\beta} A(\alpha) m^{-1}(\alpha|\beta) A(\beta)
\end{equation}
between gravity amplitudes $M_n$ and Yang-Mills partial amplitudes $A(\alpha)$. Whilst there are many presentations of the amplitudes $m(\alpha|\beta)$, the emphasis given by AHBHY, and pursued in this paper, is combinatorial. This prompts a question: can the Yang-Mills partial amplitudes $A(\alpha)$ be related to a combinatorial object? And can the field theory KLT relation, equation \eqref{ftklt}, be given a combinatorial interpretation? Speculations of this kind have been widespread ever since Bern-Cachazo-Johansson suggested that Yang-Mills partial amplitudes can be written in the form
$$
A_{YM} = \sum_{\alpha,\beta} n(\alpha)c(\beta)m(\alpha|\beta),
$$
where $c(\beta)$ are colour factors and $n(\alpha)$ are numerators that obey Jacobi-type relations. \cite{bcj08} The field theory KLT relation, equation \eqref{ftklt}, then reads
$$
M_n = \sum_{\alpha,\beta} n(\alpha)n(\beta)m(\alpha|\beta).
$$
A combinatorial presentation of these identities would perhaps be interesting, especially if it relied on $FN_n$ (discussed following equation \eqref{ppp}) which is, in some sense, dual to $\mc{M}_{0,n}(\mb{R})$.

\paragraph{Kawai-Lewellen-Tye.}
In section \ref{statement'}, we observed that the diagonal components $m_{\alpha'}(\alpha|\alpha)$ of the inverse KLT kernel are related to point counting in the cones generated by the dual associahedra $\mc{A}(\alpha)^*$. Is this appearance of $m_{\alpha'}(\alpha|\alpha)$ a fluke? Or is it related to something more interesting? To begin with, can the off-diagonal entries $m_{\alpha'}(\alpha|\beta)$ be related to the intersections of the dual associahedra? To answer this, it may be helpful to find an explicit map between $FN_n$ and the open string moduli space. Alternatively, the importance of the lattice $L$ suggests that we look for answers in toric geometry. As emphasised below, in `the toric dictionary,' our formulas for $m(\alpha|\beta)$ have a natural interpretation in terms of toric geometry. What is the toric interpretation of our formula for $m_{\alpha'}(\alpha|\alpha)$? One might hope that intersection theory on the appropriate toric variety could be used to encode the full inverse KLT kernel|though this is mere speculation. It would then be necessary to understand how our intersecting dual associahedra (or toric cycles) are related to Mizera's presentation of $m_{\alpha'}(\alpha|\beta)$ as the intersection of associahedra in $\mc{M}_{0,n}(\mb{R})$. It would be interesting to make any of this precise and to elaborate the relation of point counting to the string theory KLT relations.
 
\paragraph{The toric dictionary.}
We conclude by describing the natural correspondence between cones and toric varieties. Under this correspondence, theorem \ref{theorem} takes on a new significance: the amplitudes $m(\alpha|\beta)$ can be understood as coming from an integral over subvarieties in some toric variety. These integrals can be explicitly evaluated using the Duistermaat-Heckman (`localisation') formula. In this paragraph, we will briefly describe the toric dictionary as it bears on our results in section \ref{statement}. Some more details are included in appendix \ref{oldlattice}, but the relevant facts are as follows. Given a cone $C^*\subset \mc{K}_n^*$ and a lattice $L^* \subset \mc{K}_n^*$, a standard construction produces an associated variety $X_{C^*}$ with a toric action. The torus $T_{C^*}$ that acts on $X_{C^*}$ may be identified with the unit cell of the lattice $L$ in $\mc{K}_n$. Given these identifications, the moment map for the $T_{C^*}$ action is a map
$$
\mu : X_{C^*} \rightarrow \Lie(T_{C^*})^* = \mc{K}_n^*
$$
and the image of this map turns out to be $C^*$. This means that we can write
$$
I_{C^*}(Z,1) = \int\limits_{C^*} e^{- Z \cdot W} \d W|_{C^*} = \int\limits_{X_{C^*}} e^{- Z\cdot \mu(x)} \d x,
$$
where $\d x$ is the pushforward of $\d W$. The torus action $T_{C^*}$ has a fixed point $x_o \in X_{C^*}$. If the cone $C^*$ is based at the origin, $\mu(x_o) = 0$. The Duistermaat-Heckmann formula then evaluates the integral as
$$
\int\limits_{X_{C^*}} e^{- Z\cdot \mu(x)} \d x = \frac{1}{\det_{x_o}(-Z)},
$$
where $Z$ is regarded as a generator for a 1-parameter subgroup of $T_{C^*}$ and acts on $X_{C^*}$ at $x_o$. It turns out that this is just a toric version of the formula
$$
\int\limits_{C^*} e^{- Z \cdot W} \d W|_{C^*} = \frac{1}{\prod_{i} ( Z\cdot W_i)},
$$
where $C^* = \Cone (W_1,W_2,...)$. Things get more interesting when more than one cone is involved. A union of cones is called a fan (provided that the intersections of the cones are cones). Consider a fan formed from multiple cones $C^*,D^*,...$ in $\mc{K}_n^*$. A standard construction associates a toric variety to this fan. This is constructed by glueing together the toric varieties $X_{C^*}, X_{D^*},...$ associated to each cone. If the cones $C^*$ and $D^*$ intersect in a cone $C^*\cap D^*$, then $C^*\cap D^*$ defines a subvariety $X_{C^*\cap D^*}$ in both $X_{C^*}$ and $X_{D^*}$. We can then glue $X_{C^*}$ and $X_{D^*}$ together by identifying this subvariety. Returning to the context of section \ref{statement}, let $Y$ be the toric variety associated to the fan $\Cone (FL_n)$. (See section \ref{embed} for $FL_n$.) Then all cones $C^* \subset FN_n$ give rise to subvarieties $X_{C^*}$ in $Y$. If $C^*$ is the intersection of two other cones, then $X_{C^*}$ is also an intersection of toric subvarieties. Since all the n-point amplitudes $m(\alpha|\beta)$ can be written in terms of the integrals $I_{C^*}$ for various $C^* \subset FN_n$, we can rewrite theorem \ref{theorem} as
\begin{equation}
\label{swp}
m(\alpha|\beta) = \int\limits_{X_{\alpha\beta}} e^{-Z\cdot \mu(x)} \d x,
\end{equation}
where $X_{\alpha\beta}$ is some cycle (or toric subvariety) in $Y$.\footnote{$X_{\alpha\beta}$ is homologous to what we might write as $X_{\Cone (\p \mc{A}(\alpha)^*)}\cap X_{\Cone (\p \mc{A}(\beta)^*)}$.} This rewriting of theorem \ref{theorem} is strictly tautological, but it may suggest new ways forward. In particular, notice that the exponent in the integrand, $-Z \cdot \mu(x)$, is a (perfect) Morse function on the toric variety $Y$. (See \cite{atiyah}, for instance.) In fact, we get a whole family of Morse functions by varying $Z$. Given these observations, we might consider arriving at the amplitudes $m(\alpha|\beta)$ using a supersymmetric-quantum-mechanics model with $Y$ as its target space (in the manner of \cite{witten} or \cite{fln}). It is not clear whether such a model would be useful for understanding the double copy relation. Moreover, it remains to work out whether or not intersections in the toric variety $Y$ are related to the components $m_{\alpha'}(\alpha|\beta)$ of the inverse KLT kernel.

\appendix
\section{Review of polyhedral cones and convex polytopes}
\label{old}
This appendix reviews some facts about polyhedral cones in vector spaces and in discrete lattices. This topic is closely related to toric geometry and a fast-paced review of polyhedral cones appears in Appendix A of Oda's textbook on toric geometry. \cite{oda88} Ref. \cite{bp99} begins with a substantial review of polyhedra in lattices. I have also relied heavily on \cite{brion88}. Section \ref{em} on Euler-Maclaurin formulas presents a recent result that first appears in \cite{bv05} following earlier work in \cite{bv97a} and \cite{bv97b}. Useful lecture notes on these subjects appear on A. Barvinok's university webpage, some of which have been published in book form, \cite{barvinok08}.

\subsection{Cones}
\label{oldcones}
Let $V$ be a vector space of dimension $\dim V$. A \emph{polyhedron} is a subset $P\subset V$ defined by finitely many linear inequalities and linear equalities. In other words, a polyhedron is the intersection of some planes and half-spaces. In particular, a \emph{polytope} is a bounded polyhedron. A polytope $P$ is a \emph{convex polytope} if for all $Z,Z'\in P$, the midpoint $(Z+Z')/2$ is in $P$, too. One way to produce convex polytopes, is from their vertices. Take some vectors $Z_1,...,Z_k\in V$ as the vertices of a convex polytope $P$. Then one way to express $P$ is as the \emph{convex hull} of these points
$$
P = \Conv (Z_1,...,Z_k) = \left\{ \sum_{i=1}^k c_i Z_i \,|\, \sum_{i=1}^k c_i = 1 ~~\text{and}~~0\leq c_i \leq 1\,\forall\,i\right\}.
$$
For any polyhedron $P$, we can define the affine space containing $P$, $\Aff(P)$, as the smallest hyperplane in $V$ containing all of $P$. By the \emph{interior} of $P$, $\Int(P)$, we mean the interior of $P$ regarded as a subset of $\Aff(P)$. The \emph{boundary} of $P$ is then $\p P = P\backslash\Int(P)$. The study of polyhedra can often be reduced to a study of `polyhedral cones' using theorem \ref{conicdecomposition}, which we will come to shortly. In general, a \emph{cone} is a subset $S\subset V$ which is (i) convex, and (ii) for $Z\in P$, $\lambda Z$ is also in $P$ for all positive real numbers $\lambda$. A cone is a \emph{polyhedral cone} if it has flat sides. (i.e. if the cone is also a polyhedron in the sense described above.) One way to produce polyhedral cones is by taking the \emph{conic hull} of some vectors,
\begin{equation}
\label{defconichull}
\Cone (Z_1,...,Z_k) = \left\{ \sum_{i=1}^k c_i Z_i \,|\, 0\leq c_i \,\forall\,i\right\}.
\end{equation}
A cone is a \emph{pointed cone} if its apex is a point. This means that a pointed cone contains no fully extended line (i.e. it contains only half-lines). A polyhedron $P$ can be decomposed into cones. Let $\Ver(P)$ be the vertices of $P$. For any vertex $Z \in \Ver(P)$ define the \emph{tangent cone}
$$
\Tan_P(Z) = \{ Z+X\,|\, Z + \lambda X \in P ~~\text{for some}~~\lambda > 0 \}
$$
and the \emph{cone of directions}
\begin{equation}
\label{dirrr}
\Dir_P(Z) = \{X\,|\, Z + \lambda X \in P ~~\text{for some}~~\lambda > 0 \}.
\end{equation}
Clearly these two cones are related by a translation:
$$
\Tan_P(Z) = Z + \Dir_P(Z).
$$
A useful observation is that $P$ can be decomposed into the sum of its tangent cones.\\

\begin{thm}
\label{conicdecomposition}
For any polyhedron $P$,
$$
P \simeq \sum_{Z \in \Ver(P)} \Tan_P(Z),
$$
where $\simeq$ denotes equality modulo the addition or subtraction of polyhedra containing fully extended lines.\footnote{This is theorem 3.5 in \cite{bp99}.}\\
\end{thm}

\begin{ex}
\label{exstrip}
Consider the polyhedron $P$ in $V = \mb{R}^2$ with vertices $Z_1 = (0,0)$ and $Z_2 = (1,0)$ and tangent cones
$$
\Tan_P(Z_1) = \Cone \left(\begin{bmatrix} 0\\1\end{bmatrix},\begin{bmatrix} 1\\0\end{bmatrix}\right)
$$
and
$$
\Tan_P(Z_2) = Z_2 + \Cone \left(\begin{bmatrix} 0\\1\end{bmatrix},\begin{bmatrix} -1\\0\end{bmatrix}\right).
$$
Then the sum of these two cones is
$$
\Tan_P(Z_1) + \Tan_P(Z_2) = P + H,
$$
where $H$ is the upper-half-plane $H = \{(x,y)\in V\,|\,y\geq 0\}$. Since $H$ contains a fully extended line, we conclude that
$$
\Tan_P(Z_1) + \Tan_P(Z_2) \simeq P,
$$
as in the theorem. See figure \ref{exa2}.
\end{ex}
\begin{figure}
\begin{center}
\includegraphics[scale=0.40]{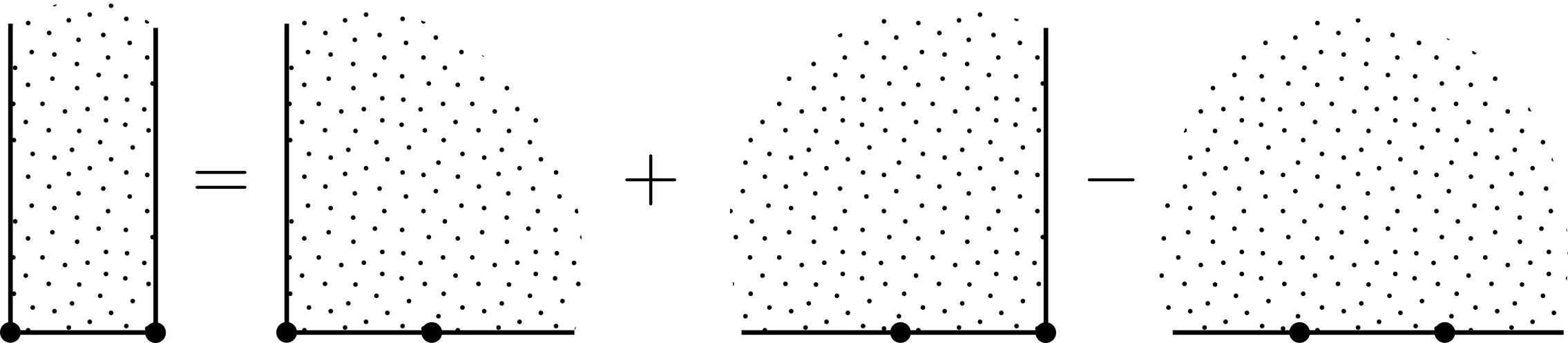}
\end{center}
\caption{An example of the polyhedral decomposition into cones, theorem \ref{conicdecomposition}.}
\label{exa2}
\end{figure}

For any set $S$ we may define its \emph{dual}
$$
S^* = \{ W \in V^*\,|\, W\cdot Z \geq -1 \, \forall\, Z \in S \}.
$$
If $C$ is a cone, this definition implies that
\begin{equation}
\label{olddualconedef}
C^* = \{ W \in V^*\,|\, W\cdot Z \geq 0 \, \forall\, Z \in C \}.
\end{equation}

We now mention an interesting result, related to the decomposition in theorem \ref{conicdecomposition}.\\

\begin{thm}
\label{dualcones}
Let $P$ be a polytope in $V$ with vertices $\Ver_P$. Then
$$
\bigcup_{Z\in \Ver_P} \Dir_P(Z)^* \simeq V^*,
$$
where $\simeq$ denotes equality as sets modulo excision by polyhedra contained in proper subspaces of $V^*$.
\end{thm}

Example \ref{exstrip} is a non-example of this theorem, since the infinite strip is not a polytope (it is not bounded). Indeed, the union of the dual cones $\Dir_P(Z_1)^*$ and $\Dir_P(Z_2)^*$ in example \ref{exstrip} only cover a half-space in $V^* = \mb{R}^2$. For a bonafide example of theorem \ref{dualcones} consider the equalitateral triangle.\\

\begin{ex}
\label{extri}
The equilateral triangle is the convex hull of three points
$$
e_1 = (0,1),\qquad e_2 = (\sqrt{3}/2,-1/2)\qquad e_3 = (-\sqrt{3}/2,-1/2).
$$
Let $T = \Conv(e_1,e_2,e_3)$ be the triangle. Then the direction cones are, for instance,
$$
\Dir_T(e_1) = \Cone(e_2-e_1,e_3-e_1).
$$
Notice that $e_3$ is orthogonal to $e_2-e_1$ and $e_2$ is orthogonal to $e_3-e_1$. Based on this, one can show that the dual cone is
$$
\Dir_T(e_1)^* = \Cone(e_2,e_3).
$$
Likewise,
$$
\Dir_T(e_2)^* = \Cone(e_3,e_1)\qquad\text{and}\qquad \Dir_T(e_3)^* = \Cone(e_1,e_2).
$$
It is clear that these three cones fill the vector space.
\end{ex}

\begin{figure}
\begin{center}
\includegraphics[scale=0.40]{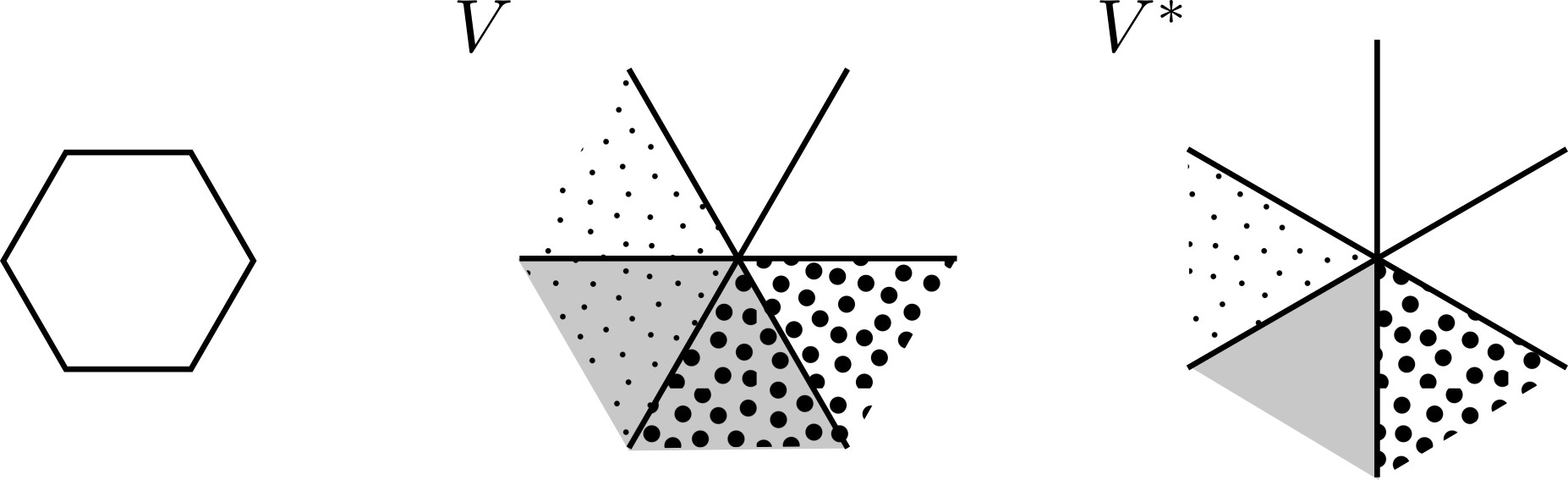}
\end{center}
\caption{Another example of theorem \ref{dualcones}.}
\label{exa5}
\end{figure}

This example appears in the text in connection with the dual associahedra for the four point amplitude. See section \ref{reviewdual}. For another example of the theorem, consider the hexagon shown in figure \ref{exa5}. The cones defined by the vertices of the hexagon overlap with each other. However, the dual cones tile the dual space in agreement with the theorem.

\subsection{Continuous valuation}
\label{oldint}
We now define a valuation (or `volume') on polyhedra, and on cones in particular, that plays a significant role in section \ref{statement} of the main text. Let $P \subset V$ be a polyhedron containing no fully extended line. For $W \in V^*$ we consider the integral
$$
I_P(W,\alpha') = \int\limits_P e^{-\alpha'W\cdot Z} \d Z_P.
$$
When this integral does not converge, we set $I_P(W,\alpha') = 0$. If the interior of $P$ has dimension $d$ (possibly smaller than $\dim V$) the measure $\d Z_P$ appearing in the integral is the $d$-dimensional Euclidean volume on the interior of $P$. Notice the following shift property,
$$
I_{Z_o+P}(W) = e^{-\alpha' W\cdot Z_o} I_P(W).
$$
We now compute an example.\\

\begin{ex}
\label{exstrip2}
Consider the cone $C = \Cone((1,0),(0,1))$, which is the first quadrant of the plane. If $W = (w_1,w_2)$ are coordinates on $V^*$,
$$
I_C(W,\alpha') = \int\limits_0^{\infty}\d a\int\limits_0^{\infty}\d b \,e^{-\alpha'aw_1-\alpha'bw_2} = (-1)^2 \frac{1}{(\alpha')^2w_1w_2}.
$$
\end{ex}

The result of the calculation in example \ref{exstrip2} is a function which is divergent as $w_1\rightarrow 0$ or $w_2\rightarrow 0$. Indeed, $w_1=0$ and $w_2= 0$ are the boundaries of the dual cone $C^*$. In general, if $C$ is a cone, then the integral
$$
I_C(W,\alpha')
$$
is well defined for all $W \in C^*$ and diverges to positive infinity as $W$ approaches the boundary of $C^*$. (See \cite{oda88} proposition A.10.) The calculation in example \ref{exstrip2} generalises to an arbitary polyhedral cone.\\

\begin{thm}
\label{intcal1}
If $C = \Cone(Z_1,...,Z_k)$ be a polyhedral cone with $k \leq \dim V$ we compute that
$$
I_C(W) = (-1)^k \frac{\Vol(\Box(Z_1,...,Z_k))}{\prod_{i=1}^kW\cdot Z_i},
$$
where $\Box(Z_1,...,Z_k)$ is the unit (open-closed) box $\{\sum_{i=1}^k c_i Z_i \,|\, 0\leq c_i < 1\,\forall\,i\}$.\\
\end{thm}

\begin{figure}
\begin{center}
\includegraphics[scale=0.40]{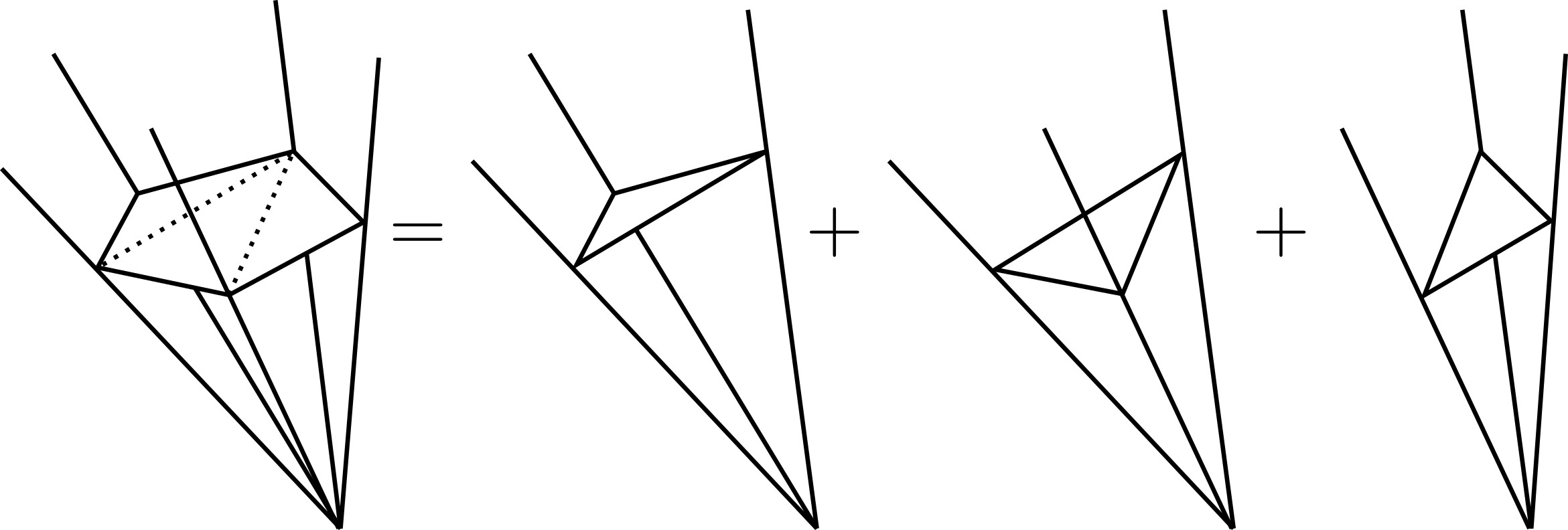}
\end{center}
\caption{Decomposing a polyhedral cone using a triangulation.}
\label{exa7}
\end{figure}
This theorem follows by computing $I$ for the standard cone $(\mb{R}^+)^k$ in $\mb{R}^k$ and then using the map from $(\mb{R}^+)^k$ to $C$. What happens when $k$ is larger than $\dim V$? In this case, we can evaluate $I_C$ for $C = \Cone(Z_1,...,Z_k)$ by giving a triangulation of the polytope $P = \Conv(Z_1,...,Z_k)$. Suppose that $I_a$ are a collection of subsets $I_a \subset \{1,...,k\}$ with length $\dim V$. Associated to each such subset is a cone
$$
C_a = \Cone \{ Z_i \,|\,i \in I_a\}.
$$
If the subsets $I_a$ define a triangulation of $P$, then
$$
C = \sum_a C_a.
$$
We sketch an example in figure \ref{exa7}. By the usual linearity of integration, we arrive at the following result.\\

\begin{thm}
\label{intcal2}
Let $C = \Cone(Z_1,...,Z_k) \subset V$ for $k > \dim V$ and suppose that some subsets $I_a \subset \{1,...,k\}$ define a triangulation of $P = \Conv(Z_1,...,Z_k)$ in $(d-1)$-simplices. Then
$$
I_C(W) = \sum_{a} I_{C_a}(W) = \sum_a (-1)^{\dim V} \frac{\Vol (\Box_a)}{\prod_{i\in I_a} W\cdot Z_i},
$$
where $\Box_a$ is the unit box generated by the $Z_i$ with $i\in I_a$.
\end{thm}

We can use the decomposition theorem, Theorem \ref{conicdecomposition}, to compute $I_P$ for any polyhedron $P$ as
\begin{equation}
\label{dhindisguise}
I_P(W,\alpha') = \sum_{Z\in \Ver(P)} I_{\Tan_P(Z)}(W,\alpha') =  \sum_{Z\in \Ver(P)} e^{-\alpha'W\cdot Z} I_{\Dir_P(Z)}(W,\alpha').
\end{equation}
The volume of $P$ is ostensibly given by the limit $\alpha'\rightarrow 0$. We can extract the volume of $P$ from $I_P(W,\alpha')$ by considering its Laurent expansion in $\alpha'$ and taking the $(\alpha')^0$ term. This is tractable because theorems \ref{intcal1} and \ref{intcal2} show us that
$$
I_{\Dir_P(Z)}(W,\alpha')
$$
is always homogeneous in $\alpha'$. If the interior of $P$ has dimension $d$, then $I_{\Dir_P(Z)}(W,\alpha')$ is homogeneous with weight $-d$.\\

\begin{thm}
\label{volthm}
Let $P$ be a polyhedron as above. The $(\alpha')^0$ term in the Laurent expansion of $I_P(W,\alpha')$ is
$$
I_P^{(0)}(W) = \sum_{Z\in \Ver(P)} \frac{1}{d!} (W\cdot Z)^d I_{\Dir_P(Z)}(W).
$$
\end{thm}

When $P$ is a polytope, this expression is the volume of $P$. However, it is also well defined when $P$ is not bounded. We give two examples.\\

\begin{ex}
\label{exstrip3}
Let $P$ be the strip defined in example \ref{exstrip}. If $W = (w_1,w_2)$ are the dual coordinates, we find
$$
I^{(0)}_P(W) = \frac{1}{2}\frac{w_1}{w_2}.
$$
So, even though $P$ has infinite Euclidean volume, $I^{(0)}_P(W)$ is still an interesting function. Indeed, the function takes values in $[0,0.5]$ with extreme values on the boundary of $P$.\\
\end{ex}

\begin{ex}
Let $T$ be the triangle introduced in example \ref{extri}. The integrals $I_{\Dir_P(Z)}(W)$ evaluate to give, for example,
$$
I_{\Dir_T(e_1)}(W) = \frac{3\sqrt{3}}{2}\frac{4}{3}\frac{1}{w_1^2-3w_2^2}.
$$
Summing the contributions in theorem \ref{volthm} gives
$$
I_T^{(0)}(W) = \frac{3\sqrt{3}}{2}.
$$
In this case we see that $I^{(0)}_T(w)$ is the Euclidean area of the triangle, which is what we expect.\\
\end{ex}

\subsection{Lattices and toric varieties}
\label{oldlattice}
Now let $L$ be a lattice in $V$. A vector $Z$ is $L$-rational if an integer multiple of it, $n Z$, is a lattice point. More generally, a subspace $W\subset V$ is $L$-rational if it is the affine space generated by points in $L$. (Equivalently, $W$ is rational if $W\cap L$ is a lattice.) A polyhedral $C$ is rational if it can be generated by a collection of lattice points. In sections \ref{appdiscrete} and \ref{em} we discuss some concrete constructions based on rational cones in a lattice|these are used in section \ref{stringyformula} of the main text to discuss a possible relation between the cones in kinematic space and the inverse KLT kernel. In this section, we will briefly recall the construction of toric varieties from cones and polyhedra. Our reason for doing this is to point out that formulas for the integral $I_P(W)$ like theorem \ref{intcal1} and equation \eqref{dhindisguise} are, in fact, disguised forms of the Duistermaat-Heckman formula. Readers not interested in this connection may skip to section \ref{appdiscrete}.

\paragraph{}
Consider the rational cone, $C = \Cone (Z_1,...,Z_k)$, where the $Z_i$ are all lattice points. Let $\Aff(C)$ be the affine space containing $C$, as above. Clearly
$$
\Aff(C) = \Span_{\mb{R}} (Z_1,...,Z_k).
$$
We can consider the lattice
$$
M = \mb{Z}[Z_1,...,Z_k] \subset \Aff (C),
$$
which is a sublattice of $L$. This lattice defines a torus
$$
T = \Aff(C) / M
$$
which we can identify with the open-closed unit box
$$
\Box(Z_1,...,Z_k) = \left\{ \sum_{i=1}^k c_i Z_i \,|\, 0\leq c_i < 1 \right\}.
$$
And the image of $L$ in $T$, which is $G = L \cap \Aff(C) / M$, is a finite subgroup of $T$. Associated to the lattice $M$ we can consider the algebra of characters. For any $Z \in M$,
$$
\chi_Z(W) = e^{2\pi i W\cdot Z}
$$
and we define the algebra $\mb{C}[M]$ to be generated by these characters with the multiplication $\chi_Z\cdot \chi_{Z'} = \chi_{Z+Z'}$. Let $L \cap \Aff(C)$ have generators $\mb{Z}[\hat Z _1,...,\hat Z_k]$. Then the algebra associated to $L \cap\Aff(C)$ is generated by the characters $T_i^{\pm 1} = \chi_{\pm \hat Z_i}$. So,
$$
\mb{C}[L\cap \Aff(C)] = \mb{C}[T_1^{\pm 1},...,T_k^{\pm 1}].
$$
The characters of $M$, $\chi_Z$ with $Z\in M$, are positive integer powers of the $T_i^{\pm 1}$. So we have an inclusion
$$
\mb{C}[M] \hookrightarrow \mb{C}[T_1^{\pm 1},...,T_k^{\pm 1}].
$$
The variety $\Spec \mb{C}[M]$ is the toric variety associated to the cone $C^*$. (For an introduction to these ideas, see Fulton \cite{fulton93}, chapter one.)\\

\begin{ex}
\label{extor}
If $C^* = \Cone (e_1,e_2)$ in $V^* = \mb{R}^3$ with lattice $L = \mb{Z}[e_1,e_2,e_3]$, the dual cone is $C = \Cone(e_1,e_2,e_3,-e_3)$. The algebra of characters of $M = \mb{Z}_{\geq 0}[e_1,e_2,\pm e_3]$ is
$$
\mb{C}[M] = \mb{C}(T_1,T_2,T_3,T_3^{-1}).
$$
The associated toric variety is
$$
X = \Spec\,\mb{C}[M] = \mb{C}\times\mb{C}\times\mb{C}^*.
$$
The torus $T = V / M$ acts on $X$ in the obvious way. There are no fixed points. Each subcone $F \subset C$ is associated with a torus embedded in $X$. And, in particular, the points $\{0\}$ is associated with the torus $\mb{C}^*\times\mb{C}^*\times\mb{C}^* \subset X$.\\
\end{ex}

\begin{ex}
\label{extor2}
If instead we consider $C^* = \Cone (e_1,e_2,e_3)$ in $V^* = \mb{R}^3$, the associated toric variety is
$$
X = \Spec\,\mb{C}[M] = \mb{C}\times\mb{C}\times\mb{C}.
$$
This has a single fixed point, namely $x_o = (0,0,0) \in X$. In general, if the generators of $C^*$ span $V^*$, then $X$ is a non-singular variety of the form $\mb{C}^{\times n}$ and there is a single fixed point, the origin.\\
\end{ex}

The reason we have recalled the construction of toric varieties is that, in this context, the formula we obtained for $I_C(W)$, theorem \ref{intcal1}, is an example of the Duistermaat-Heckman formula. We will give a terse explanation of this. The key observation is made in \cite{brion88}. For more on Duistermaat-Heckman, the standard reference is \cite{dh82}. Let $C$ be a cone with $\dim V$ generators in $V$ and let $X$ be the associated toric variety. We can regard the torus $T$ which acts on $X$ as an open box $\Box \subset V$. So the Lie algebra of $T$ is $V$ and the moment map is a map
$$
\mu : X \rightarrow \Lie(T)^* = V^*.
$$
Then, for $Z \in V$, Duistermaat-Heckman evaluates the integral
\begin{equation}
\label{dh}
\int\limits_X e^{Z\cdot \mu(x)}\d x = e^{Z\cdot\mu(x_o)}\frac{1}{\det_{x_o} Z},
\end{equation}
where $\det_{x_o} Z$ is the determinant of the action of $Z \in V$ on $X$ at the fixed point $x_o$. By using the moment map, the integral on the left-hand-side can be identified with
\begin{equation}
\label{identint}
\int\limits_X e^{Z\cdot \mu(x)}\d x = \int\limits_{C^*} e^{Z\cdot W} \d W|_{C^*}.
\end{equation}
On the other hand, the moment map sends $x_o$ to the vertex of the cone, $W_o$. We then identify the right-hand-side as
$$
e^{Z\cdot\mu(x_o)}\frac{1}{\det_{x_o} Z} = e^{Z\cdot W_o} \frac{1}{\prod (-W\cdot Z_i)}.
$$
In general, a toric variety can be associated to any `fan' of cones. (A fan is a union of cones whose faces meet each other to form sub-cones, and so on.) This is a standard construction and involves glueing together the toric varieties associated to each cone. See \cite{fulton93}, chapter 1. In particular, we can consider the fan of cones $F_P$ associated to a polytope $P$. The fan $F_P$ is the union of all dual cones $\Dir_P(Z)^*$ for vertices $Z\in\Ver_P$. Let $X$ be the associated toric variety. Then the Duistermaat-Heckman formula reads
\begin{equation}
\label{dh}
\int\limits_X e^{W\cdot \mu(x)}\d x = \sum_{\text{fixed points}}e^{W\cdot\mu(x)}\frac{1}{\det_x W}.
\end{equation}
Remarkably, this is just the formula, equation \eqref{dhindisguise}, that we obtained earlier for $I_P(W,\alpha')$ with $\alpha'$ set to $1$. Brion makes this observation in \cite{brion88}, though he may not have been the first. The correspondence between the two formulas is roughly as follows. For each of the dual cones $\Dir_P(Z)^*$,  there is a single fixed point $x_Z \in X$ of the torus action on $X$ (just as, for a single cone, there was a single fixed point). The moment map $\mu$ maps $x_Z$ to the vertex $Z \in V$. This identifies the right-hand-side of equation \eqref{dh} with equation \eqref{dhindisguise}. We will not elaborate these ideas in any detail here,---but notice that equation \eqref{identint}, for instance, gives another interpretation to the results in section \ref{statement} described in the main text.

\subsection{Discrete valuation}
\label{appdiscrete}
Given a rational cone $C$ in $V$ with a lattice $L$, we will define a function $S_C(W)$ on $V^*$. By analogy with the definition of $I_C(W)$ in section \ref{oldint}, consider the sum
$$
S_C(W,\alpha') = \sum_{Z\in C \cap L} e^{-\alpha' W\cdot Z}.
$$
This is well defined when $W$ is in the dual cone $C^*$. Otherwise, when $W$ is not in the dual cone, the summation is not well defined and we set $S_C(W)$ to zero. Notice, as for $I_C(W)$, the translation property
$$
S_{Z_o+C}(W,\alpha') = e^{-\alpha' W\cdot Z_o} S_C(W,\alpha').
$$
We can evaluate $S_C$ explicitly in the following simple example.
\\
\begin{ex}
\label{exquad2}
Let $C = \Cone((1,0),(0,1))$ be the first quadrant of the plane, as in example \ref{exstrip2}. Take $L$ to be the standard lattice $L = \mb{Z}^2$. Then writing $W = (w_1,w_2)$ for dual coordinates,
$$
S_C(W) = \sum_{a=0}^{\infty}\sum_{b=0}^{\infty} e^{-a w_1 - b w_2} = \frac{1}{(1-e^{-w_1})(1-e^{-w_2})},
$$
provided that $e^{w_1},e^{w_2} < 1$.
\end{ex}

Notice that, in example \ref{exquad2}, $e^{w_1}$ and $e^{w_2}$ are less than $1$ for $w_1,w_2>0$. That is, the sum $S_C(W)$ is defined precisely for $W$ in the dual cone $C^*$ which is the first quadrant of the plane. At the boundary, $w_1=0$ or $w_2=0$, the sum diverges. The calculation in example \ref{exquad2} is easy to generalise. The following is the lattice analog of Theorem \ref{intcal1}.\\

\begin{thm}
Let $C = \Cone(Z_1,...,Z_k)$ for $k\leq \dim V$. Then
$$
S_C(W,\alpha') = \left( \sum_{\Box \cap L} e^{-\alpha'W\cdot Z} \right) \prod_{i=1}^k \frac{1}{1 - e^{-\alpha' W\cdot Z_i}},
$$
where $\Box$ is the unit open-closed box generated by the $Z_i$.
\end{thm}

Just as $I_P$ can be used to compute polytope volumes, $S_P$ can be used to count lattice points on the interior of a polytope. Consider the Laurent expansion of $S_P(W,\alpha')$ around $\alpha' = 0$. Then the zero-order term $S_P^{(0)}(W)$ in the Laurent series is ostensibly the number of interior lattice points,
$$
S_P^{(0)}(W) = \# (P\cap L).
$$
This is strictly true when $P$ is a polytope. When $P$ is not bounded, the right-hand-side is no longer defined, but the left hand side may still be defined.\\

\begin{ex}
Let's count the number of lattice points in the triangle $T$, introduced in example \ref{extri}. The decomposition into cones gives a sum of terms of the form
$$
S_{e_1+ \Dir_T(e_1)}(W,\alpha') = \frac{e^{-\alpha' w_1}}{(1-e^{-\alpha'(w_2-w_1)})(1-e^{-\alpha'(w_3-w_1)})},
$$
where $w_i = W\cdot e_i$. Since $e_1+e_2+e_3 = 0$, we likewise have $w_1+w_2+w_3 = 0$. Summing these gives
$$
S_T^{(0)}(W) = 1.
$$
Indeed, $T$ contains one point of the lattice $L = \mb{Z}[e_1,e_2,e_3]$ in its interior.\\
\end{ex}

\begin{ex}
\label{exstrip5}
Let's return to the infinite strip $P$ as defined in example \ref{exstrip}. We have
\begin{align*}
S_P(W,\alpha') & = S_{\Dir_P(Z_1)} + e^{-\alpha' W\cdot Z_2} S_{\Dir_P(Z_2)} \\ & = \frac{1}{(1-e^{-\alpha' W_1})(1-e^{-\alpha' W_2})} +  \frac{e^{-\alpha' W_2}}{(1-e^{\alpha' W_1})(1-e^{-\alpha' W_2})}.
\end{align*}
Using the Laurent series
$$
\frac{1}{1-e^{\tau}} = - \frac{1}{\tau} + \frac{1}{2} + ...,
$$
we find that the zero-order term in the Laurent series is
$$
S_P^{(0)}(w) = 0.
$$
\end{ex}

\subsection{Generalised Euler-Maclaurin formulas}
\label{em}
How are the valuations $S_P$ and $I_P$ related to each other? In one dimension, this question is answered by the Euler-Maclaurin formula. Recall the result that, up to a remainder term,
$$
\sum_{n=a}^{b} f(n) = \int\limits_a^b \d x f(x) - \sum_{n=1}\frac{b_n}{n!}f^{(n-1)}(a) + \sum_{n=1}(-1)^n \frac{b_n}{n!} f^{(n-1)}(b).
$$
In our case, we replace $f(x)$ with the exponential $\exp(-\alpha'wz)$. We can present the sums explicitly by making use of the generating function
$$
\sum_{m=0} \frac{b_n}{n!}t^n = \frac{t}{1-e^{-t}}.
$$
Then
\begin{equation}
\label{pleasant}
\sum_{n=a}^{b} e^{-\alpha'wz} - \int\limits_a^b \d x e^{-\alpha'wz} = - \left(\frac{1}{e^{\alpha'w}-1} - \frac{1}{\alpha' w}\right)e^{\alpha' w a} - \left(\frac{1}{e^{-\alpha'w}-1} - \frac{1}{-\alpha'w}\right)e^{\alpha'w b}.
\end{equation}
This curious looking formula has a generalisation to polytopes in higher dimensions. The generalisation is given by Brion, Vergne, and Berline, and their result shows that the valuation $S_C$, for $C$ a cone, can be expressed in terms of integrals $I_F$, where the $F$ are subcones of $C$. We present their result as the following theorem. For brevity, we temporarily omit $\alpha'$ from our formulas, which is equivalent to setting $\alpha' = 1$.\\
\begin{thm}
There exists some measure $\mu$ on rational polyhedra such that
\begin{equation}
\label{bv}
S_P (W) = \sum_Q \mu_{P/Q}(W) I(P\cap Q)(W),
\end{equation}
where the sum is over all rational subspaces $Q$ which are tangent to one of the faces (of any dimension) of $P$.\\
\end{thm}

Notice that a quotient vector space $P/Q$ inherits a lattice $\hat L$ from the lattice $L\subset V$. This is given by the projection of $L$.\footnote{The lattice $\hat L$ is not to be identified with $W^{\perp}\cap L$, which is usually a strict sublattice of $\hat L$.} It is in this sense that $P/Q$ is a rational polyhedron. Berline-Vergne define $\mu_C(W)$ on cones $C$ inductively, beginning with
$$
\mu_{\{0\}}(W) \equiv 1.
$$
Putting $C = \Cone (Z_1)$ we impose equation \eqref{bv} to find
$$
S_C(W) = I_C(W) + \mu_C(W).
$$
We thus infer that
$$
\mu_C(W) = \frac{1}{1-e^{-W\cdot Z_1}} + \frac{1}{W\cdot Z_1}.
$$
Indeed, this is precisely what we already discovered from the Euler-Maclaurin formula|see equation \eqref{pleasant}. Putting $C = \Cone(Z_1,Z_2)$ in \eqref{bv} we find
$$
\mu_C(W) = S_C(W) + \mu_{\Cone(\hat Z_1)}(W) \frac{1}{W\cdot Z_1} + \mu_{\Cone(\hat Z_2)}(W) \frac{1}{W\cdot Z_2} - \frac{1}{W\cdot Z_1W\cdot Z_2},
$$
where
$$
\hat Z_1 = Z_1 - \frac{Z_1\cdot Z_2}{Z_2\cdot Z_2} Z_2\qquad\text{and}\qquad \hat Z_2 = Z_2 - \frac{Z_1\cdot Z_2}{Z_1\cdot Z_1} Z_1.
$$
Clearly, this construction could be continued and be used to define $\mu_C(W)$ for all cones $C$.

\paragraph{Relevance for amplitudes.}
In the context of our results in sections \ref{statement} and \ref{statement'}, these generalised Euler-Maclaurin formulas imply that the diagonal entries $m_{\alpha'}(\alpha|\alpha)$ of the inverse KLT kernel have an exact expansion in terms of the amplitudes $m(\alpha|\beta)$ at finite values of $\alpha'$. The coefficients in this expansion, given by $\mu$, are complicated expressions, but they can be determined algorithmically and satisfy nice properties ($\mu$ is a valuation on polytopes). It is not clear whether this expansion is of any interest physically.



\bibliography{bbl}
\bibliographystyle{utphys}

\end{document}

%% file: 1801.tex
\usepackage{graphicx}

\usepackage[dvipsnames]{xcolor}
\usepackage{amssymb}				
\usepackage{mathrsfs}				
\usepackage{dsfont}					
\newcommand{\mc}[1]{\mathcal{ #1 }}

\newcommand{\mb}[1]{\mathbb{ #1 }}
\newcommand{\mf}[1]{\mathfrak{ #1 }}

\usepackage{tikz}
\usetikzlibrary{matrix,arrows}
\usepackage{amsmath} 	
\usepackage{amsthm} 	
\usepackage{mathtools}	

\newtheorem{thm}{Theorem}[section]
\newtheorem{cor}[thm]{Corollary}
\newtheorem{lem}[thm]{Lemma}
\newtheorem{ex}[thm]{Example}
\newtheorem{con}[thm]{Conjecture}

\theoremstyle{definition}

\theoremstyle{remark}

\usepackage{framed}

\let\oldenumerate\enumerate
\renewcommand{\enumerate}{
  \oldenumerate
  \setlength{\itemsep}{2pt}
  \setlength{\parskip}{2pt}
}

\renewcommand{\d}{\text{d}}
\newcommand{\p}{\partial}

\newcommand{\Lie}{\text{Lie}}
\newcommand{\Conv}{\text{Conv}}
\newcommand{\Cone}{\text{Cone}}
\newcommand{\Ver}{\text{Vert}}
\newcommand{\Dir}{\text{Dir}}
\newcommand{\Tan}{\text{Tan}}
\newcommand{\Vol}{\text{Vol}}
\newcommand{\Aff}{\text{Aff}}
\newcommand{\Spec}{\text{Spec}}
\newcommand{\Int}{\text{Int}}
\newcommand{\Val}{\text{Val}}

\newcommand{\Span}{\text{Span}}

\renewcommand{\Box}{\text{Box}}

\newcommand{\avg}[1]{\left< #1 \right>} 	

%% file: 1801a.bbl
\providecommand{\href}[2]{#2}\begingroup\raggedright\begin{thebibliography}{10}

\bibitem{ahbhy17}
N.~Arkani-Hamed, Y.~Bai, S.~He, and G.~Yan, ``Scattering {Forms} and the
  {Positive} {Geometry} of {Kinematics}, {Color} and the {Worldsheet},'' {\em
  arXiv:1711.09102 [hep-th]} (Nov., 2017) .
  \url{http://arxiv.org/abs/1711.09102}. arXiv: 1711.09102.

\bibitem{chy14}
F.~Cachazo, S.~He, and E.~Y. Yuan, ``Scattering of {Massless} {Particles}:
  {Scalars}, {Gluons} and {Gravitons},''
  \href{http://dx.doi.org/10.1007/JHEP07(2014)033}{{\em Journal of High Energy
  Physics} {\bfseries 2014} no.~7, (July, 2014) }.
  \url{http://arxiv.org/abs/1309.0885}. arXiv: 1309.0885.

\bibitem{dg14}
L.~Dolan and P.~Goddard, ``Proof of the {Formula} of {Cachazo}, {He} and {Yuan}
  for {Yang}-{Mills} {Tree} {Amplitudes} in {Arbitrary} {Dimension},''
  \href{http://dx.doi.org/10.1007/JHEP05(2014)010}{{\em Journal of High Energy
  Physics} {\bfseries 2014} no.~5, (May, 2014) }.
  \url{http://arxiv.org/abs/1311.5200}. arXiv: 1311.5200.

\bibitem{mizera1711}
S.~Mizera, ``Scattering {Amplitudes} from {Intersection} {Theory},'' {\em
  arXiv:1711.00469 [hep-th, physics:math-ph]} (Nov., 2017) .
  \url{http://arxiv.org/abs/1711.00469}. arXiv: 1711.00469.

\bibitem{matsumoto98}
K.~Matsumoto, ``Intersection numbers for logarithmic \$k\$-forms,'' {\em Osaka
  Journal of Mathematics} {\bfseries 35} no.~4, (1998) 873--893.
  \url{https://projecteuclid.org/euclid.ojm/1200788347}.

\bibitem{csz15}
C.~Ceballos, F.~Santos, and G.~M. Ziegler, ``Many non-equivalent realizations
  of the associahedron,''
  \href{http://dx.doi.org/10.1007/s00493-014-2959-9}{{\em Combinatorica}
  {\bfseries 35} no.~5, (Oct., 2015) 513--551}.
  \url{https://link.springer.com/article/10.1007/s00493-014-2959-9}.

\bibitem{ahbl17}
N.~Arkani-Hamed, Y.~Bai, and T.~Lam, ``Positive {Geometries} and {Canonical}
  {Forms},'' \href{http://dx.doi.org/10.1007/JHEP11(2017)039}{{\em Journal of
  High Energy Physics} {\bfseries 2017} no.~11, (Nov., 2017) }.
  \url{http://arxiv.org/abs/1703.04541}. arXiv: 1703.04541.

\bibitem{bcj08}
Z.~Bern, J.~J.~M. Carrasco, and H.~Johansson, ``New {Relations} for
  {Gauge}-{Theory} {Amplitudes},''
  \href{http://dx.doi.org/10.1103/PhysRevD.78.085011}{{\em Physical Review D}
  {\bfseries 78} no.~8, (Oct., 2008) }. \url{http://arxiv.org/abs/0805.3993}.
  arXiv: 0805.3993.

\bibitem{st14}
S.~Stieberger and T.~R. Taylor, ``Closed {String} {Amplitudes} as
  {Single}-{Valued} {Open} {String} {Amplitudes},''
  \href{http://dx.doi.org/10.1016/j.nuclphysb.2014.02.005}{{\em Nuclear Physics
  B} {\bfseries 881} (Apr., 2014) 269--287}.
  \url{http://arxiv.org/abs/1401.1218}. arXiv: 1401.1218.

\bibitem{ms17}
C.~R. Mafra and O.~Schlotterer, ``Non-abelian \${Z}\$-theory: {Berends}-{Giele}
  recursion for the \${\textbackslash}alpha'\$-expansion of disk integrals,''
  \href{http://dx.doi.org/10.1007/JHEP01(2017)031}{{\em Journal of High Energy
  Physics} {\bfseries 2017} no.~1, (Jan., 2017) }.
  \url{http://arxiv.org/abs/1609.07078}. arXiv: 1609.07078.

\bibitem{klt}
H.~Kawai, D.~C. Lewellen, and S.~H.~H. Tye, ``A relation between tree
  amplitudes of closed and open strings,''
  \href{http://dx.doi.org/10.1016/0550-3213(86)90362-7}{{\em Nuclear Physics B}
  {\bfseries 269} no.~1, (May, 1986) 1--23}.
  \url{http://www.sciencedirect.com/science/article/pii/0550321386903627}.

\bibitem{mizera1706}
S.~Mizera, ``Inverse of the string theory {KLT} kernel,''
  \href{http://dx.doi.org/10.1007/JHEP06(2017)084}{{\em Journal of High Energy
  Physics} {\bfseries 2017} no.~6, (June, 2017) 84}.
  \url{https://link.springer.com/article/10.1007/JHEP06(2017)084}.

\bibitem{devadoss}
S.~L. Devadoss, ``Tessellations of {Moduli} {Spaces} and the {Mosaic}
  {Operad},'' {\em arXiv:math/9807010} (July, 1998) .
  \url{http://arxiv.org/abs/math/9807010}. arXiv: math/9807010.

\bibitem{mizera1708}
S.~Mizera, ``Combinatorics and {Topology} of {Kawai}-{Lewellen}-{Tye}
  {Relations},'' \href{http://dx.doi.org/10.1007/JHEP08(2017)097}{{\em Journal
  of High Energy Physics} {\bfseries 2017} no.~8, (Aug., 2017) }.
  \url{http://arxiv.org/abs/1706.08527}. arXiv: 1706.08527.

\bibitem{ky94}
M.~Kita and M.~Yoshida, ``Intersection {Theory} for {Twisted} {Cycles} {II} -
  {Degenerate} {Arrangements},''
  \href{http://dx.doi.org/10.1002/mana.19941680111}{{\em Mathematische
  Nachrichten} {\bfseries 168} no.~1, (Jan., 1994) 171--190}.
  \url{http://onlinelibrary.wiley.com/doi/10.1002/mana.19941680111/abstract}.

\bibitem{kapranov93}
M.~M. Kapranov, ``The permutoassociahedron, {Mac} {Lane}'s coherence theorem
  and asymptotic zones for the {KZ} equation,''.

\bibitem{zr94}
V.~Reiner and G.~M. Ziegler, ``Coxeter-associahedra,''
  \href{http://dx.doi.org/10.1112/S0025579300007452}{{\em Mathematika}
  {\bfseries 41} no.~2, (Dec., 1994) 364--393}.
  \url{https://www.cambridge.org/core/journals/mathematika/article/coxeterassociahedra/B16F7B95B252E50E38BD30C1A96690F7}.

\bibitem{postnikov05}
A.~Postnikov, ``Permutohedra, associahedra, and beyond,'' {\em
  arXiv:math/0507163} (July, 2005) . \url{http://arxiv.org/abs/math/0507163}.
  arXiv: math/0507163.

\bibitem{lmh}
A.~Postnikov, V.~Reiner, and L.~Williams, ``Faces of {Generalized}
  {Permutohedra},'' {\em arXiv:math/0609184} (Sept., 2006) .
  \url{http://arxiv.org/abs/math/0609184}. arXiv: math/0609184.

\bibitem{early}
N.~Early, ``Canonical {Bases} for {Permutohedral} {Plates},'' {\em
  arXiv:1712.08520 [hep-th]} (Dec., 2017) .
  \url{http://arxiv.org/abs/1712.08520}. arXiv: 1712.08520.

\bibitem{atiyah}
M.~F. Atiyah, ``Convexity and {Commuting} {Hamiltonians},''
  \href{http://dx.doi.org/10.1112/blms/14.1.1}{{\em Bulletin of the London
  Mathematical Society} {\bfseries 14} no.~1, (Jan., 1982) 1--15}.
  \url{http://onlinelibrary.wiley.com/doi/10.1112/blms/14.1.1/abstract}.

\bibitem{witten}
E.~Witten, ``A {New} {Look} {At} {The} {Path} {Integral} {Of} {Quantum}
  {Mechanics},'' {\em arXiv:1009.6032 [hep-th]} (Sept., 2010) .
  \url{http://arxiv.org/abs/1009.6032}. arXiv: 1009.6032.

\bibitem{fln}
E.~Frenkel, A.~Losev, and N.~Nekrasov, ``Instantons beyond topological theory
  {I},''. \url{https://arxiv.org/abs/hep-th/0610149}.

\bibitem{oda88}
T.~Oda, {\em Convex {Bodies} and {Algebraic} {Geometry}: {An} {Introduction} to
  the {Theory} of {Toric} {Varieties}}.
\newblock Springer-Verlag, Place of publication not identified, softcover
  reprint of the original 1st ed. 1988 edition~ed., Jan., 1988.

\bibitem{bp99}
A.~Barvinok and J.~E. Pommersheim, ``An {Algorithmic} {Theory} of {Lattice}
  {Points} in {Polyhedra},''.
  \url{http://citeseerx.ist.psu.edu/viewdoc/summary?doi=10.1.1.170.9496}.

\bibitem{brion88}
M.~Brion, ``Points entiers dans les polyèdres convexes,'' {\em Annales
  scientifiques de l'École Normale Supérieure} {\bfseries 21} no.~4, (1988)
  653--663. \url{https://eudml.org/doc/82241}.

\bibitem{bv05}
N.~Berline and M.~Vergne, ``Local {Euler}-{Maclaurin} formula for polytopes,''
  {\em arXiv:math/0507256} (July, 2005) .
  \url{http://arxiv.org/abs/math/0507256}. arXiv: math/0507256.

\bibitem{bv97a}
M.~Brion and M.~Vergne, ``Residue formulae, vector partition functions and
  lattice points in rational polytopes,''
  \href{http://dx.doi.org/10.1090/S0894-0347-97-00242-7}{{\em Journal of the
  American Mathematical Society} {\bfseries 10} no.~4, (1997) 797--833}.
  \url{http://www.ams.org/jams/1997-10-04/S0894-0347-97-00242-7/}.

\bibitem{bv97b}
M.~Brion and M.~Vergne, ``Lattice {Points} in {Simple} {Polytopes},'' {\em
  Journal of the American Mathematical Society} {\bfseries 10} no.~2, (1997)
  371--392. \url{http://www.jstor.org/stable/2152855}.

\bibitem{barvinok08}
A.~Barvinok, {\em Integer {Points} in {Polyhedra}}.
\newblock Amer Mathematical Society, Zürich, Aug., 2008.

\bibitem{fulton93}
W.~Fulton, {\em Introduction to {Toric} {Varieties}.}
\newblock Princeton University Press, Princeton, N.J, July, 1993.

\bibitem{dh82}
J.~J. Duistermaat and G.~J. Heckman, ``On the {Variation} in the cohomology of
  the symplectic form of the reduced phase space,''
  \href{http://dx.doi.org/10.1007/BF01399506}{{\em Invent.Math.} {\bfseries 69}
  (1982) 259--268}.

\end{thebibliography}\endgroup
